\documentclass{IEEEtaes}

\jvol{XX}
\jnum{XX}
\jmonth{XXXXX}
\paper{1234567}
\pubyear{2025}
\doiinfo{TAES.2025.Doi Number}

\usepackage[T1]{fontenc}
\usepackage[utf8]{inputenc}

\usepackage[english]{babel}

\usepackage{microtype}

\usepackage{amsmath,amssymb,amsfonts}
\usepackage{mathtools}        
\usepackage{bm}               
\usepackage{amsthm}           

\newtheorem{theorem}{Theorem}
\newtheorem{lemma}{Lemma}
\newtheorem{proposition}{Proposition}
\newtheorem{corollary}{Corollary}
\theoremstyle{definition}
\newtheorem{definition}{Definition}
\newtheorem{assumption}{Assumption}
\theoremstyle{remark}

\allowdisplaybreaks

\usepackage{graphicx}
\usepackage{wrapfig}
\usepackage{booktabs}         
\usepackage{multirow}
\usepackage{array}
\usepackage{makecell}
\usepackage{tabularx}         
\usepackage{float}            

\usepackage[caption=false,font=footnotesize]{subfig}

\usepackage{siunitx}          
\usepackage{xspace}
\usepackage{braket}           

\usepackage{cite}             

\usepackage[nameinlink,noabbrev]{cleveref}

\usepackage{algorithm}
\usepackage{algpseudocode}    
\usepackage{listings}         

\lstdefinestyle{mystyle}{
  basicstyle=\ttfamily\footnotesize,
  numbers=left, numberstyle=\tiny,
  stepnumber=1, numbersep=5pt,
  breaklines=true, tabsize=2, showstringspaces=false
}
\usepackage{tikz}
\usetikzlibrary{arrows.meta,positioning,calc,fit,decorations.pathreplacing}

\usepackage{xcolor}
\usepackage{soul}             
\usepackage{comment}          
\usepackage{verbatim}         
\usepackage{xr}               
\usepackage{url}              
\usepackage{etoolbox}         

\usepackage[disable]{todonotes}

\newcommand{\nat}{\mathbb{N}}
\newcommand{\reals}{\mathbb{R}}

\newcommand{\dataset}{\mathcal{D}}
\newcommand{\ndataset}{\tilde{\mathcal{D}}}

\newcommand{\CX}{\mathcal{X}}

\newcommand{\gaussN}{\mathcal{N}}

\newcommand{\Prob}{\mathbb{P}}

\newcommand{\Tr}{\textrm{Tr}}

\newcommand{\optstat}{\Phi^*(\bar{\boldsymbol{\beta}})}
\newcommand{\optstati}{\hat{\Phi}^i(\ndataset)}
\newcommand{\stati}{\Phi^i}

\newcommand{\fcdf}{F_{\rvtest}}
\newcommand{\fccdf}{\bar{F}_{\rvtest}}
\newcommand{\Popt}{\mathbb{P}_{\optstat}}
\newcommand{\barstati}{\bar{\Phi}^i(\ndataset)}
\newcommand{\baroptstat}{\bar{\Phi}^*(\ndataset)}
\newcommand{\rvtest}{\Psi}
\newcommand{\rvtesti}{\Psi^i}

\newcommand{\numagents}{M}
\newcommand{\boldf}{\boldsymbol{f}}
\newcommand{\conset}{X_c^{\dtime}}
\newcommand{\effset}{X_E}
\newcommand{\simplex}{\mathcal{W}_{\numagents}}
\newcommand{\psimplex}{\simplex^+}
\newcommand{\woptset}{S(\sw)}

\newcommand{\lconst}{\alpha_{\dtime}}

\newcommand{\dtime}{t}
\newcommand{\respi}{\var^i_{\dtime}}
\newcommand{\nrespi}{\tilde{\var}^i_{\dtime}}
\newcommand{\noise}{\epsilon}
\newcommand{\noisei}{\noise^i}
\newcommand{\ndistr}{\Lambda}

\newcommand{\sw}{\mu}

\newcommand{\var}{\beta}
\newcommand{\varti}{\var_t^i}

\newcommand{\dimn}{N}

\newcommand{\statei}{x}
\newcommand{\measi}{y^i}
\newcommand{\snoisei}{w}
\newcommand{\mnoisei}{v^i}
\newcommand{\stime}{t} 
\newcommand{\ftime}{k} 
\newcommand{\sdim}{q} 
\newcommand{\mdim}{p} 
\newcommand{\sncov}{Q_{\stime}(\lconst)}

\newcommand{\argo}{\beta}
\newcommand{\na}{M}

\newcommand{\kstate}{\hat{x}} 
\newcommand{\kcov}{\Sigma} 
\newcommand{\ARE}{\mathcal{A}(\lconst,\varti,\kcov)}

\newcommand{\hu}{\hat{u}}
\newcommand{\hlam}{\hat{\lambda}}
\newcommand{\hb}{\hat{\beta}}

\newcommand{\wass}{\mathcal{W}}

\newcommand{\CE}{\mathbb{E}}
\newcommand{\CV}{\mathcal{V}}
\newcommand{\bv}{\boldsymbol{v}}

\newcommand{\util}{f}

\usepackage{tikz}
\usetikzlibrary{arrows.meta,positioning}
\usetikzlibrary{calc}

\begin{document}

\title{Multi-Agent Inverse RL for Identifying Pareto-Efficient Coordination -- Distributionally Robust Approach} 

\author{Luke Snow}
\affil{Department of Electrical \& Computer Engineering, Cornell University, Ithaca, NY, USA} 

\author{Vikram Krishnamurthy}
\member{Fellow, IEEE}
\affil{Department of Electrical \& Computer Engineering, Cornell University, Ithaca, NY, USA }

\receiveddate{ This research was supported by the Army Research Office grant W911NF-24-1-0083 and National Science Foundation grant CCF-2312198.}

\authoraddress{Author emails are as follows. Luke Snow: las474@cornell.edu, Vikram Krishnamurthy: vikramk@cornell.edu}

\maketitle

\begin{abstract} Multi-agent inverse reinforcement learning (IRL) aims to identify Pareto-efficient behavior in a multi-agent system, and reconstruct utility functions of the individual agents. Motivated by the problem of detecting UAV coordination, how can we construct a statistical detector for Pareto-efficient behavior given noisy measurements of the decisions of a multi-agent system?
This paper approaches this IRL problem by deriving necessary and sufficient conditions for a dataset of multi-agent system dynamics to be consistent with Pareto-efficient coordination, and providing algorithms for recovering utility functions which are consistent with the system dynamics. We derive an optimal statistical detector for determining Pareto-efficient coordination from noisy system measurements, which minimizes Type-I statistical detection error. Then, we provide a utility estimation algorithm which minimizes the worst-case estimation error over a statistical ambiguity set centered at empirical observations; this min-max solution achieves \textit{distributionally robust} IRL, which is crucial in adversarial strategic interactions. We illustrate these results in a detailed example for detecting Pareto-efficient coordination among multiple UAVs given noisy measurement recorded at a radar. We then reconstruct the  utility functions of the UAVs in a distributionally robust sense.
\end{abstract}


\section{Introduction}

In strategic environments, autonomous dynamical systems such as UAVs are now  ubiquitous for reconnaissance, surveillance, and combative purposes. Often such autonomous systems are deployed in groups, e.g., UAV swarms, in order to collect information more efficiently or to multiply the combative force \cite{zhou2020uav}, \cite{nawaz2021uav} . Furthermore, these multi-agent intelligent systems typically have sophisticated sensors and communication capabilities which allow them to respond in real-time to an adversary's probe, e.g., radar tracking signals \cite{pace2009detecting}. This results in a strategic interaction between the multi-agent system and the adversary; the study of this interaction at the physical layer, for instance analyzing electromagnetic suppression techniques, is referred to as \textit{electronic warfare} \cite{spezio2002electronic}. 

We consider a multi-agent strategic interaction scenario between a sensor and a multi-agent system. Multi-agent inverse reinforcement learning (IRL) aims to identify Pareto-efficient behavior in a multi-agent system, and reconstruct utility functions of the individual agents. We take the perspective of an analyst observing the behavior of the system, and address two  questions: (i) How can we detect \textit{Pareto-efficient coordination} in the multi-agent system? (ii) How can we {\em reconstruct} individual utility functions which induce the observed aggregate behavior? 
Addressing these questions  allows us to understand the functionality of the multi-agent system, and also robustly predict its future behavior.\footnote{This paper substantially extends our preliminary work  in conference papers \cite{snow2022identifying} and \cite{snow2023statistical}. Specifically, the current paper  extends from the setting of Pareto-efficiency among cognitive radars to the full generality of a multi-agent system; we also reveal how our techniques can apply to a Pareto-efficient multi-UAV system. This work also derives a novel distributionally robust utility reconstruction procedure in this general setting, which is absent in the previous works. Furthermore, we provide extended discussion, motivation and numerical simulations.}

\begin{figure}
\centering
  \includegraphics[width=0.8\linewidth,scale=0.25]{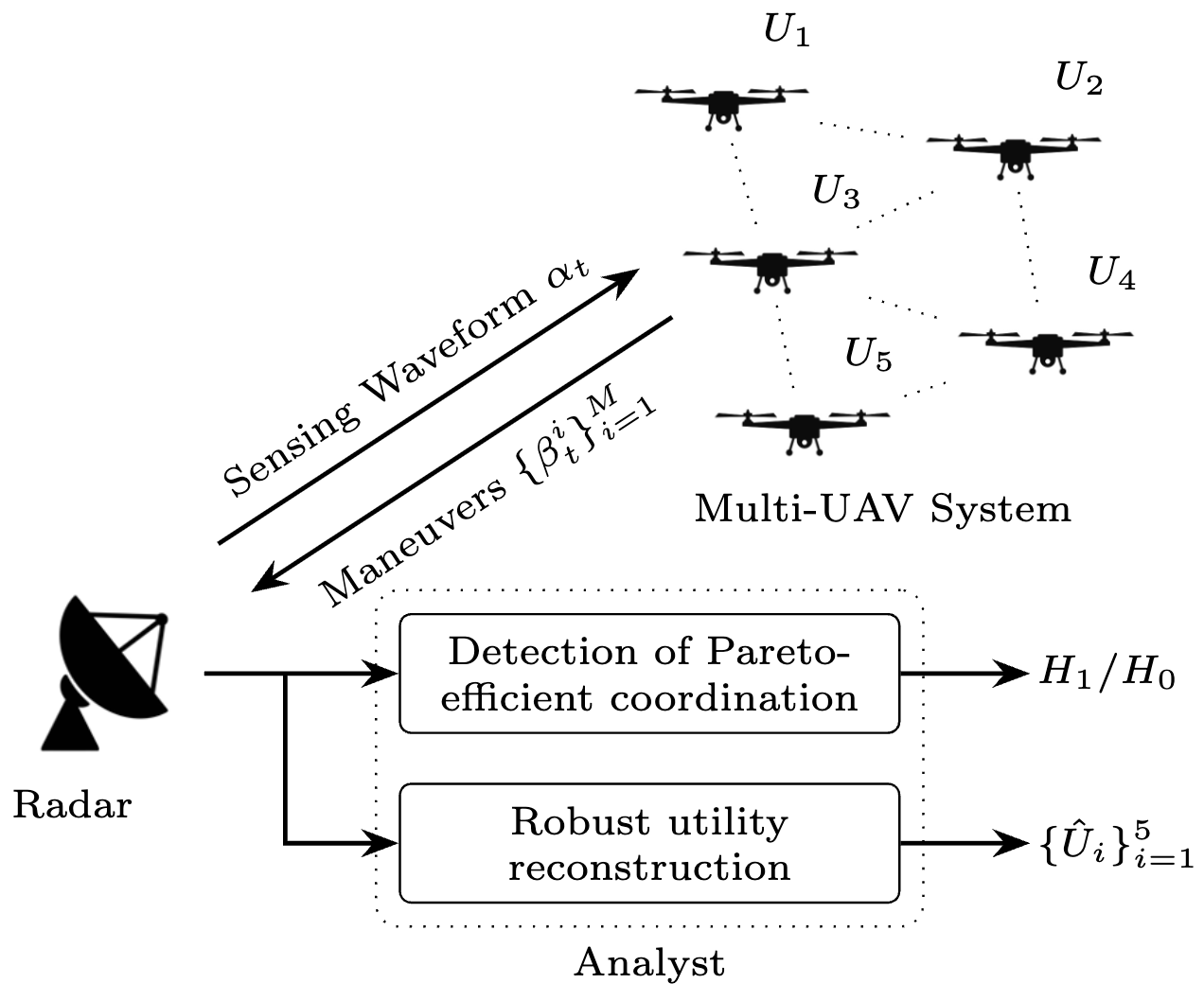}
  \caption{Radar -- Multi-UAV Interaction. We represent the high-level radar tracking waveform (parameters) by $\alpha_t$, and the target network maneuvers by $\{\beta_t^i\}_{i\in[M]}$. The analyst observes the dataset $\dataset = \{\alpha_t, \{\beta_t^i\}_{i=1}^{\na}, t\in[0,T]\}$ and has two goals. First, to determine whether the dataset is consistent with Pareto-efficient coordination. Second, to produce distributionally robust estimates of the utility functions driving the system.}
  \label{fig:interaction}
\end{figure}

We study this problem at a higher level of abstraction than traditional electronic warfare investigations; this allows us to formulate the Pareto-efficient coordination problem as a general linearly-constrained multi-objective optimization. Then, the problem of detecting Pareto-efficient coordination and reconstructing feasible objective functions becomes that of \textit{inverse multi-objective optimization}, which we approach using tools from microeconomic theory. Specifically, the theory of revealed preferences \cite{varian2006revealed}, \cite{chambers2016revealed} provides a principled methodology for approaching this inverse multi-objective optimization. We also show how this framework arises naturally from physical-layer considerations such as radar waveform modulation and multi-target filtering algorithms, in an extended example on multi-UAV Pareto-efficient coordination detection. Figure~\ref{fig:interaction} illustrates the radar–multi-UAV interaction: an analyst observes sensor data and system dynamics to test for coordination and reconstruct utility functions. This is made possible by an equivalence between our microeconomic revealed preference framework and multi-target filtering dynamics (notably the joint probabilistic data association filter), providing an analytical bridge to advanced tracking algorithms\footnote{This analytical bridge is detailed in Section~\ref{sec:uavcoord}.}. However, our revealed preference framework itself is \textit{independent} of the radar–UAV setting and applies \textit{broadly} to multi-agent systems.


\subsection{Main Results and Organization}

Our main contributions are summarized as follows:

\begin{itemize}
    \item[-] \textit{Step 1A: Conditions for Testing Pareto-Efficiency}. 
    Our first result in multi-agent IRL provides necessary and sufficient conditions for a dataset of observed system dynamics to be \textit{consistent} with Pareto-efficiency. From these conditions, we also derive a method for constructing utility functions with respect to which the observed data is Pareto-efficient, thereby generalizing Afriat's seminal theorem in microeconomic revealed preference theory \cite{afriat1967construction} to the multi-agent regime. If the dataset is not consistent with Pareto-efficiency, we quantify its proximity to such satisfaction, i.e., “how close” the system is to Pareto-efficient coordination. This near-optimality testing is crucial for our subsequent developments in detection and robust utility reconstruction. 
    
    \item[-] \textit{Step 1B: Optimal Statistical Detection}. 
    The conditions of Step 1A apply to deterministic data. To handle noisy or suboptimal observations, we derive a statistical detector for Pareto-efficiency hypothesis testing, and prove several optimality guarantees, including minimization of Type-I error. Beyond detection, the conditions naturally yield a method for predicting future system behavior.
    
    \item[-] \textit{Step 2: Distributionally Robust Utility Estimation}. 
    When Pareto-efficient behavior is measured with noise or is only approximately optimal, classical reconstruction methods fail to provide accurate utilities. We develop a distributionally robust utility estimation algorithm that minimizes the worst-case reconstruction error. In adversarial settings, bounding this worst-case error is critical since a single misleading estimate may compromise detection, tracking, or prediction. We use  techniques in distributionally robust optimization \cite{rahimian2019distributionally}, \cite{rahimian2022frameworks}, \cite{kuhn2019wasserstein}, which has emerged as the state-of-the-art framework for obtaining robust solutions under statistical ambiguity. To our knowledge, this is the first application of distributionally robust optimization to achieve reliable multi-agent IRL. 
\end{itemize}

{\em Organization}. 
Section~\ref{sec:MOO} provides background on multi-objective optimization and Pareto efficiency in multi-agent systems. 
Section~\ref{sec:MOO_detector} develops the necessary and sufficient conditions for Pareto-efficient coordination and the optimal statistical detector. 
Section~\ref{sec:robust} presents our distributionally robust utility reconstruction procedure. 
Section~\ref{sec:uavcoord} applies our framework to a multi-UAV coordination problem using radar tracking signals, demonstrating its relation to waveform modulation and filtering algorithms. 
Section~\ref{sec:numeric} provides illustrative numerical simulations.

\vspace{-0.2cm}

\subsection{Related Works and Context}

\textit{Why Pareto-Efficiency as Model for Coordination?}
Pareto-efficiency is a widely utilized formalization for coordination in multi-agent systems \cite{snow2023statistical}, \cite{wang2020thirty}, \cite{wismans2014pruning}, \cite{liao2022energy}, \cite{sabino2018topology}, \cite{gao2018multi}, \cite{gomez2020pareto}, \cite{marden2014achieving}, \cite{ruadulescu2020multi} since it captures optimal resource allocation under heterogeneous objectives. In multi-agent systems, agents face shared constraints (e.g., power, bandwidth, detectability) while pursuing distinct goals. Pareto-efficiency guarantees \emph{global rationality}: no agent can improve without harming another. This makes it the weakest possible useful notion of collective optimality—less restrictive than requiring agreement on a common utility, yet strong enough to exclude inefficient or wasteful joint actions. This in turn motivates statistical detection of Pareto efficient coordination.


\textit{Adversarial Intent Detection via Revealed Preferences.} This work extends a line of research in using revealed preference techniques for learning in adversarial sensing settings \cite{krishnamurthy2020identifying}, \cite{pattanayak2022meta}, \cite{pattanayak2022can}. The methods introduced in this paper extend these techniques to detect \textit{Pareto-efficiency in multi-agent systems}, and to reconstruct utility functions in a \textit{distributionally robust} manner. In short, these contributions enable effective \textit{group-intent recognition} using these tools. 

\textit{Role of Deep Neural Network (DNN) Methods.}
We focus  on detecting Pareto efficient behavior and reconstruction of utility functions with provable performance guarantees (necessary and sufficient conditions). 
Recent works  that use DNN  architectures for reward learning from dynamical system observations (see \cite{wulfmeier2015maximum}, \cite{wulfmeier2017large})  can be used in conjunction with our methods in two ways \footnote{A companion paper \cite{zhang2025infer} studies group intent as the outcome of a cooperative game that modulates the probabilities in the target dynamics. The paper then estimates the utility using graph based neural networks.}: First,
deep auto-encoders can be used as a pre-processing step  to map a dataset to a lower dimensional feature
dataset. Second,  DNNs can be  used as a functional approximator to learn a utility. 
However, note 
that our proposed methods operate effectively on  small datasets of system dynamics, 
whereas DNN  methods often require  large datasets to be effective.

\section{Background. Constrained Multi-Objective Optimization in Coordinated Systems}
\label{sec:MOO}
Since our goal is to detect Pareto-efficient coordination, in this section
 we provide a brief background to  multi-objective optimization and Pareto-efficiency.
 
\subsection{Multi-Objective Optimization}

We  define a {\em coordinating} multi-agent system as one which  satisfies  multi-objective optimality, where each agent $i$ has a distinct objective function $f^i$. This is quantified as the following constrained multi-objective optimization:

\begin{align}
\begin{split}
\label{eq:MOP}
 &\arg \max _\var \{f^1(\var),\dots,f^{\numagents}(\var)\}\ \\& s.t. \ \var \in \conset := \{\var \in \reals^n  : \alpha_t^\top \var \leq 1 \}
 \end{split}
\end{align}
where the linear constraint\footnote{For vector $x$ we let $x^\top $ represent the transpose of $x$.} $\alpha_t^\top  \var$ is bounded by 1 without loss of generality (see Sec. I-A of \cite{krishnamurthy2020identifying}). In single-objective optimization, the goal is to find the best feasible argument which maximizes the objective. However, in multi-objective optimization there will seldom exist an argument $\var$ which simultaneously maximizes all objectives, i.e., there will be trade-offs between objectives for varying argument $\var$. 

A fundamental solution concept for the multi-objective optimization problem \eqref{eq:MOP} is that of Pareto efficiency:
\begin{definition}[Pareto-Efficiency]
\label{def:par_opt}
 For fixed $\{\{f^i(\cdot)\}_{i=1}^{\numagents},\lconst\}$ and a vector $\var^*\in \conset$, let 
\begin{align*}
\begin{split}
    &Z^t(\var^*) = \{\var \in \conset : f^i(\var) \geq f^i(\var^*) \ \forall i \in [\numagents]\} \\
    &Y^t(\var^*) = \{\var \in \conset : \exists k : f^k(\var) > f^k(\var^*) \}
\end{split}
\end{align*}
The vector $\var^*$ is said to be Pareto-efficient if 
\begin{equation}
\label{efficiency}
Z^t(\var^*) \cap Y^t(\var^*) = \emptyset
\end{equation}
where $\emptyset$ is the empty set. Intuitively, a vector is Pareto-efficient if there does not exist another vector in the feasible set $\conset$ which increases the value of some objective $f^i(\cdot)$ without simultaneously decreasing the value of some other objective $f^j(\cdot)$, $i,j\in [\numagents]$.
\end{definition}
We then denote the set of all Pareto-efficient solutions to the problem \eqref{eq:MOP} as 
\begin{equation}
\label{eq:effset}
\effset(\{f^i\}_{i=1}^M, \lconst) := \{\var^* \in \conset : \eqref{efficiency} \textrm{ is satisfied}\}
\end{equation}
and we say that $\var^*$ solves \eqref{eq:MOP} if and only if $\var^*$ is Pareto-efficient, i.e.,
\begin{align*}
   &\var^* \in \{ \arg \max_\var \{f^1(\var),\dots,f^{\numagents}(\var)\}\ s.t. \ \var \in \conset \}\\
   &\Longleftrightarrow \ \var^* \in \effset(\{f^i\}_{i=1}^M, \lconst)
\end{align*}

$\effset(\{f^i\}_{i=1}^M, \lconst)$ is referred to as the "Pareto-frontier", and is a hypersurface in the ambient space with dimensionality $M$.
Denoting $\boldf(\var) = (f^1(\var),\dots,f^{\numagents}(\var))^\top $, we can use the following problem of weighted sum (PWS) \cite{gass1955computational} to obtain a Pareto-efficient solution:
\begin{equation}
\label{eq:PWS}
\max \ \sw^\top \boldf(\var) \ s.t. \ \var \in \conset
\end{equation}
where $\sw = (\sw^1,\dots,\sw^{\numagents})^\top  \in \reals^{\numagents}_+$. The set of weights $\sw$ is restricted to the non-negative unit simplex, denoted as $\simplex := \{\sw \in \reals^{\numagents}_+ : \boldsymbol{1}^\top \sw = 1\}$. Denote the set of optimal solutions for \eqref{eq:PWS} as
\begin{equation}
\label{eq:woptset}
\woptset = \arg \max_\var\{\sw\top \boldf(\var) : \var \in \conset\}
\end{equation}
Then, letting $\psimplex = \{\sw \in \reals^{\numagents}_{++} : \boldsymbol{1}\top \sw = 1\}$ denote the unit simplex with each weight $\sw^i$ strictly positive, we have \cite{miettinen1999nonlinear}:

\begin{equation}
\label{eq:effrel}
\bigcup_{\sw \in \psimplex}\woptset \subseteq \effset(\{f^i\}_{i=1}^M, \lconst) \subseteq \bigcup_{\sw \in \simplex} \woptset
\end{equation}
where the right-most inclusion holds only if the objective functions are concave and the feasible set is convex \cite{miettinen1999nonlinear}. In words, all optimal solutions $S(\mu)$ for $\mu\in\psimplex$ are Pareto-efficient, and under concave objective functions and convex constraint sets all Pareto-efficient solutions can be obtained by solving $S(\mu)$ with some $\mu\in\simplex$.

Figure~\ref{fig:pareto} illustrates the concept of Pareto-efficiency in the objective space $(f_1(\beta),f_2(\beta))$. The shaded region corresponds to the feasible image $\boldf(X_t^c)$ defined by the constraint set $X_t^c=\{\beta:\alpha_t^\top\beta \leq 1\}$. Any point in the interior is \emph{dominated}: meaning there exists another feasible allocation which increases one objective without reducing the other. In contrast, points on the upper-right boundary (the Pareto frontier) satisfy Definition~\ref{def:par_opt}: no feasible $\beta$ can improve one objective $f^i(\beta)$ without decreasing another $f^j(\beta)$. Weighted-sum optimization~\eqref{eq:woptset} selects among these frontier points, with different choices of $\mu$ tracing out the set of efficient allocations. 
\begin{figure}[t]
\centering
\begin{tikzpicture}[>=Stealth,scale=0.5]

\draw[->] (0,0) -- (8.6,0) node[below,right] {$f_1(\beta)$};
\draw[->] (0,0) -- (0,8.6) node[left] {$f_2(\beta)$};

\def\feas{(0.9,1.1) .. controls (3.2,0.6) and (6.4,1.7) .. (7.25,3.2)
          .. controls (7.6,4.4) and (5.9,6.3) .. (3.6,7.5)
          .. controls (2.3,7.8) and (1.1,6.5) .. (0.9,1.1)}
\fill[blue!7] \feas -- cycle;
\node[blue!60!black] at (3.2,5.25) {\small Feasible image $\boldf(X_t^{c})$};

\draw[line width=1.4pt,red!70!black]
      (7.25,3.2) .. controls (7.6,4.4) and (5.9,6.3) .. (3.6,7.5)
      node[pos=0.5,above right,xshift=0.6ex] {\small Pareto frontier};

\draw[->,gray!70,thick] (1.0,0.8) -- (3,2.6);
\node[gray!70,above right] at (2.2,1.6) {\scriptsize better};

\coordinate (A) at (3,4);
\filldraw[black] (A) circle (2pt);
\node[above right] at (A) {\small dominated};

\coordinate (B) at (5.4,6.3); 
\coordinate (C) at (4.45,7.0); 
\coordinate (D) at (7.225,4); 
\filldraw[black] (B) circle (2.2pt);
\filldraw[black] (C) circle (2.2pt);
\filldraw[black] (D) circle (2.2pt);
\node[above right] at (B) {\small efficient};
\node[above right] at (C) {\small efficient};
\node[above right] at (D) {\small efficient};


\foreach \x in {2,4,6,8} \draw[gray!30] (\x,0) -- (\x,-0.08) node[below,black!60] {\scriptsize \x};
\foreach \y in {2,4,6,8} \draw[gray!30] (0,\y) -- (-0.08,\y) node[left,black!60] {\scriptsize \y};

\end{tikzpicture}
\caption{Pareto-efficiency in objective space for a 2-objective optimization. The interior (blue) is the feasible image $f(X_t^c)$. Points on the upper boundary (red) are Pareto-efficient: improving one objective necessarily worsens the other. Black points on the frontier illustrate efficient solutions, while any points in the interior are dominated. The  aim of this paper is to \textit{detect} Pareto-efficiency given a dataset from a  multi-agent system.}
\label{fig:pareto}
\end{figure}
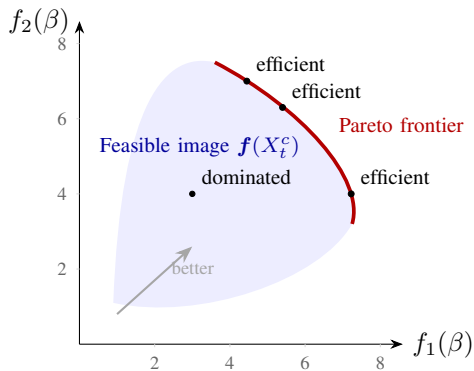

Next we will detail how Pareto-efficiency can be defined as the optimality condition of a coordinated multi-agent system. We then turn to our methodology for \textit{detecting} coordination from observed system behavior.

\subsection{Coordinated Multi-Agent Systems}
\label{sec:radarnet}
Here we specify a general dynamical systems model of our multi-agent system. We detail how one can derive a constrained multi-objective optimization as the natural mathematical formalism defining multi-agent coordination.

We consider the interaction between a sensor and a multi-agent system comprising $\na$ heterogeneous agents. The agents evolve according to a state-space model, and the sensor records noisy observations of each agent's state.


\begin{definition}[Multi-Agent System]
\label{inter_dynamics}
We introduce the following state-space sensing dynamics:
\begin{align}
    \begin{split}
    \label{eq:state_dyn}
        \textrm{agent $i$ state}: x_{t} &\in \reals^{\mdim}, \
        x_t \sim p_{\varti}(x|x_{t-1}) \\
        \textrm{agent $i$ parameter} : \varti &\in \reals^{\dimn}_+, \, i\in[\na]\\
        \textrm{sensor output}: \lconst &\in \reals^{\dimn}_+ \\
        \textrm{sensor measurement}: \measi_{\ftime} &\in \reals^{\mdim}, \
        \measi_{\ftime} \sim p_{\lconst}(y|x_{\ftime}^i) \\
    \end{split}
\end{align}
\end{definition}
\noindent 
Here $[x]$ denotes the set $\{1,\dots,x\}$. Each agent $i$ has utility function $\util^i: \reals^{\dimn}_{+} \mapsto \reals$, quantifying its objective, and may adjust its dynamics through parameter $\beta_t^i$ to achieve its objective. In a \textit{coordinated} multi-agent system, the individual dynamics $\beta_t^i$ are coupled, and the group behaves in order to maximize the aggregate utility:
\begin{definition}[Coordinated System]
\label{def:coord}
We define a coordinating system to be a group of $\na$ agents, each with individual concave, continuous and monotone increasing objective functions\footnote{This structure is widely assumed, and is natural in microeconomic and electronic warfare settings \cite{krishnamurthy2020identifying}} $\util^i: \reals^{\dimn} \to \reals, i\in[\na]$, which act according to dynamics parametrizations $\{\varti\}_{i=1}^{\na}$ such that\footnote{The constraint bound 1 is without loss of generality, see \cite{krishnamurthy2020identifying}.}
 \begin{gather}
\begin{aligned}
\label{def:coord_eq}
    \begin{split}
        \{\varti\}_{i=1}^{\na} \in \arg\max_{\{\argo^i\}_{i=1}^{\na}} \sum_{i=1}^{\na}\mu^i\util^i(\argo^i) \,\, s.t. \,\, \alpha_t^\top  (\sum_{i=1}^{\na} \argo^i) \leq 1
    \end{split}
\end{aligned}\raisetag{2.4\baselineskip}\end{gather}
for a set of weights $\mu^i > 0$.
\end{definition}

Observe that \eqref{def:coord_eq} is exactly 
$\{\beta_t^i\}_{i=1}^{\na} \in S(\mu)$ for constraint set $X_c^t = \{\{\beta_i\}_{i=1}^{\na} : \alpha_t^\top \sum_{i=1}^{\na}\beta_i \leq 1\}$, and thus \eqref{def:coord_eq} is a special case of \eqref{eq:MOP}.  
Thus, a group which acts according to \eqref{def:coord_eq} optimally (is \textit{Pareto-efficient}) parametrizes its state kernels $p_{\beta_t^i}(y|x_t)$ subject to each objective function, the sensing dynamics of the sensor, and a constraint on the joint agent dynamics. This joint constraint $\alpha_t^\top (\sum_{i=1}^{\na}\beta^i) \leq 1$ arises naturally from physical-layer considerations of, e.g., cognitive radar network power constraints \cite{snow2023statistical} or UAV network detectability. Indeed, we will illustrate the latter in detail in an extended example in Section~\ref{sec:uavcoord}. For now we simply motivate this constraint as a way of coupling the agent dynamics within a joint optimization. 

\textit{Detecting Pareto-Efficiency}: In the following sections we take the perspective of an analyst who observes the behavior of the system. We construct algorithms to determine if the system coordinates, that is, if its dynamics satisfy \eqref{def:coord_eq}. We provide necessary and sufficient conditions for a dataset of observations to be consistent with such Pareto-efficiency, and derive an optimal statistical detection scheme for achieving this in noise. Furthermore, in Section~\ref{sec:robust} we provide a distributionally robust technique for estimation of the utility functions $\{f^i\}_{i=1}^{\na}$. 


\subsection{Detecting Pareto Efficiency.  Examples of Multi-Agent Systems}

The methods that we will develop in subsequent sections  apply to detecting Pareto efficiency for \textit{any} multi-agent coordinated system which has dynamics \eqref{eq:state_dyn} and can be formalized as performing the jointly constrained optimization \eqref{def:coord_eq}.
In Section~\ref{sec:uavcoord} we will  provide an extended example of how the dynamics of a coordinated multi-UAV system naturally gives rise to the abstracted mathematical form of \eqref{def:coord_eq}.  We now  briefly discuss  examples of other multi-agent systems fitting this framework, where an analyst may wish to determine whether observed behavior is consistent with coordination. 

\textit{Detecting Pareto-Efficient Spectrum Allocation.} In multi-cell 5G/6G systems, base stations allocate spectrum and transmit power while balancing throughput, latency, and fairness \cite{wang2020thirty}. An analyst observing network scheduling or beamforming could test whether these allocations are Pareto-efficient subject to spectrum interference constraints $\sum_i \alpha_t^\top  \beta_i \leq 1$, and, upon confirmation, reconstruct implicit utility weights that reveal operator priorities.

\textit{Autonomous Vehicle Teaming.} Vehicles in a team jointly adjust velocity and spacing to minimize energy use, maximize throughput, and ensure safety \cite{wang2020cav}. From traffic flow data, one may detect whether the platoon acts in a coordinated, Pareto-efficient manner under safety-distance constraints, and recover utility functions that quantify the relative emphasis on efficiency versus safety.

\textit{Smart Grid Dispatch.} Distributed generators and storage units schedule electricity production subject to supply-demand balance and transmission limits \cite{revathi2025multi}. Observing dispatch decisions, an analyst could test for coordination (Pareto-efficient resource allocation) and reconstruct cost or renewable-integration utilities, enabling prediction of grid responses to shocks or adversarial interventions.


Thus, the methodology developed here applies beyond our illustrative example of multi-UAV systems, to any system where multiple agents optimize distinct utilities under a shared linear constraint.

\begin{figure*}[t]
\centering
\begin{tikzpicture}[
  >=Latex,
  node distance=1.3cm and 1.6cm,
  box/.style={draw, rounded corners, align=center, inner sep=4pt, fill=gray!4, font=\scriptsize}
]

\node[box, minimum width=2.3cm] (probe) {Sensor Probe \\ $\lconst$};
\node[box, right=of probe, minimum width=3.4cm] (system)
  {Multi-agent system \\ (coordination) \\[2pt]
   $\{\respi\}\in\arg\max \sum_i \sw^i U^i(\beta_i)$ \\
   s.t. $\alpha_t^\top \sum_{i=1}^{\na}\beta_i \leq 1$};
\node[box, right=of system, minimum width=2.7cm] (analyst)
  {Analyst observes \\ $(\lconst,\{\respi\})$};
\node[box, right=of analyst, minimum width=3.2cm] (LP)
  {LP feasibility test \\ eq.~\eqref{af_ineq}};

\draw[->, thick] (probe) -- (system) node[midway,above] {\scriptsize constraint};
\draw[->, thick] (system) -- (analyst);
\draw[->, thick] (analyst) -- (LP) node[midway,above] {\scriptsize dataset $\dataset$};

\end{tikzpicture}
\caption{Probe–response revealed preference framework. The environment issues a probe $\lconst$, 
the system responds optimally with $\{\respi\}$, and the analyst observes $(\lconst,\{\respi\})$. 
Testing coordination reduces to checking LP feasibility \eqref{af_ineq}.}
\label{fig:inter_dynam}
\end{figure*}
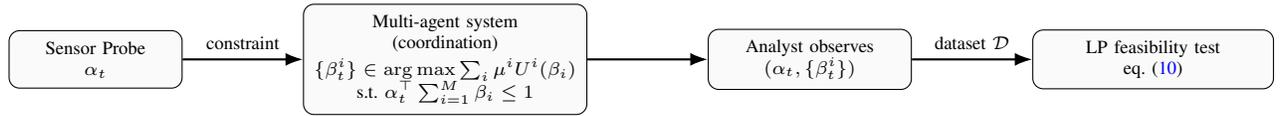

\section{Result 1. Statistical Detection of Pareto-Efficient Coordination} 
\label{sec:MOO_detector}


Recall our main results in multi-agent IRL constitute a two-step procedure: detecting Pareto efficient coordination  in a multi-agent system and subsequently reconstructing the utilities of the agents. This section details the first step; we provide an optimal statistical detector for determining if a noisy dataset of  dynamics from a multi-agent system are consistent with Pareto-efficient coordination. We first provide an equivalence between \textit{deterministic} multi-agent system coordination (Def.~\ref{def:coord}) and the feasibility of a testable linear program; allowing us to test coordination by solving this linear program, from deterministic system observations. Corollary~\ref{cor:Utrec} provides a mechanism for reconstructing feasible utility functions from this \textit{deterministic} data, which allows for understanding the goals of the system and predicting future dynamics. We then extend these results to the stochastic regime by providing a statistical detector for determining whether \textit{noisy} measurements of system dynamics are consistent with multi-objective optimization (coordination). \footnote{If we concretize to a radar-UAV interaction, we assume the target can observe the signals $\{\varti, t\in[T]\}_{i=1}^{\numagents}$ through e.g., an omni-directional receiver or by radar tracking signals. See \cite{pace2009detecting} for physical-layer considerations of radar waveform observation, detection, and classification.}


\subsection{Main Result 1A. Linear Programming Formulation for Detection of Deterministic Pareto-Efficient Coordination.} 

We model the interaction between an evolving constraint (or probe) $\alpha_t$ and the resulting multi-agent system response $\{\beta_t^i\}_{i=1}^M$. 
If the agents are coordinating, each probe–response pair $(\alpha_t,\{\beta_t^i\}_{i=1}^M)$ arises from a linearly constrained multi-objective optimization \eqref{def:coord_eq}, thereby revealing the agents’ underlying objectives. 
The analyst observes a dataset \[\mathcal{D}=\{\,\alpha_t,\{\beta_t^i\}_{i=1}^M : t\in[T]\}\] and seeks to determine whether it is \emph{rationalizable}, i.e., whether there exist utility functions $\{U^i\}_{i=1}^M$ under which all observations are consistent with Pareto-efficient coordination. 
This corresponds to the classical revealed preference problem \cite{afriat1967construction,mcfadden2006revealed,cherchye2011revealed}, and leads to our first main result: a linear program that tests rationalizability and reconstructs consistent utilities, which we later extend to noisy observations via an optimal statistical detector. 

\subsubsection{Testing Pareto-Efficiency}
Specifically, here we provide a necessary and sufficient condition for the dataset $\dataset$ to be consistent with multi-objective optimization.
\begin{theorem}
 \label{thm:cherchye1}
    Let $\dataset$ be a set of observations. The following are equivalent:
    \begin{enumerate}
    \item there exist a set of $M$ concave and continuous objective functions $U^1,\dots,U^m$ and weights $\sw \in \psimplex$ such that $\forall t \in [T]$:
    \begin{align}
    \begin{split}
    \label{thm1:rat}
        \{\respi\}_{i=1}^{\numagents} \in &\arg\max_{\{\argo^i\}_{i=1}^{\numagents}} \sum_{i=1}^{\numagents} \sw^i U^i(\argo^i) \ \  \\& s.t. \ \alpha_t^\top  (\sum_{i=1}^{\numagents}\argo^i ) \leq 1
    \end{split}
    \end{align}
    \item there exist numbers $u_j^i > 0, \lambda_j^i > 0$ such that for all $s,t \in [T]$, $i \in [M]$: 
    \begin{equation}
    \label{af_ineq}
        u_s^i - u_t^i - \lambda_t^i\alpha_t^\top [\var_s^i - \varti] \leq 0
    \end{equation}
    \end{enumerate}
\end{theorem}
\begin{proof}
See Theorem 1 of \cite{snow2022identifying}
\end{proof}

The equivalence in Theorem~\ref{thm:cherchye1} provides a tractable test of coordination
rationalizability. Rather than searching over unknown utilities to determine consistency, one can check the 
feasibility of a system of linear inequalities. This is the multi-agent generalization 
of Afriat’s theorem \cite{afriat1967construction}, which is a seminal result in the theory of microeconomic revealed preferences. Afriat’s theorem provides a nonparametric constructive test to identify if an agent is a
budget-constrained utility maximizer. A remarkable feature of Afriat’s theorem is that if the
dataset can be rationalized by a utility function, then it can be rationalized by a continuous,
concave, monotonic utility function. That is, violations of continuity, concavity, or monotonicity
cannot be detected with a finite number of observations.

\subsubsection{Reconstructing Rationalizing Utility Functions}
Given feasibility (coordination), we can use the following Corollary to reconstruct objective functions which rationalize the observed responses. 

\begin{corollary}
\label{cor:Utrec}
Given constants $u_t^i, \lambda_t^i, t\in[T],i\in[M]$ which make \eqref{af_ineq} feasible, explicit monotone and continuous objective functions that "rationalize" the dataset \\ $\{\lconst, \respi, t \in [T], i \in [M]\}$ are given by
\begin{equation}
\label{eq:Utrec}
     U^i(\cdot) = \min_{t \in [T]} \left[u_t^i + \lambda_t^i\alpha_t^\top [\cdot - \respi] \right]
\end{equation}
i.e., \eqref{thm1:rat} is satisfied with objective functions \eqref{eq:Utrec}.
\end{corollary}
\begin{proof}
See Lemma 1 of \cite{snow2022identifying}.
\end{proof}

This Corollary gives us a way to predict future system outputs by reconstructing feasible utility functions which generate the multi-objective optimal dynamics. Then, for probe $\alpha_t$, Pareto-efficient responses can be predicted by computing the constrained multi-objective optimal solutions \eqref{def:coord_eq}.  \footnote{The reconstructed piecewise linear utilities \eqref{eq:Utrec} are not unique and are ordinal by construction. Ordinal means that any positive monotone increasing transformation of the utility function
will also rationalize the dataset $\dataset$. This is why the budget constraint $\alpha_t^\top \sum_{i=1}^{\na}\beta_t^i$ is without loss
of generality; it can be scaled by an arbitrary positive constant and Theorem~\ref{thm:cherchye1} still holds.}


\subsubsection{Predicting Future Responses}
Having reconstructed rationalizing utilities via Corollary~\ref{cor:Utrec}, we can predict how the multi-agent system will respond to unseen probes. 
Given a new constraint vector $\alpha_{t+1}$, the predicted Pareto-efficient responses $\{\beta_{t+1}^i\}_{i=1}^M$ are obtained by solving the same multi-objective program \eqref{def:coord_eq}, but now with the reconstructed utilities. 

\begin{algorithm}[H]
\caption{Predicting Future System Dynamics}
\begin{algorithmic}[1]
    \State \textbf{Input:} Dataset $\dataset = \{\alpha_t, \{\beta_t^i\}_{i=1}^{\na}, t\in[T]\}$
    \If{\eqref{af_ineq} feasible}
    reconstruct utilities $\{U^i(\cdot)\}_{i=1}^M$ via \eqref{eq:Utrec}. 
    \State \textit{Predict Responses}: Given probe $\alpha_{t+1}$, solve the following via a convex optimization routine
    \begin{align*}
        \{\beta_{t+1}^i\}_{i=1}^M &\in \arg\max_{\{\argo^i\}_{i=1}^M} 
        \sum_{i=1}^M \mu^i U^i(\argo^i) 
        \quad \\&s.t.\ \alpha_{t+1}^\top  \Big(\sum_{i=1}^M \argo^i\Big) \leq 1
    \end{align*}
    \State Return the Pareto-frontier $\effset(\{U^i\}_{i=1}^M, \alpha_{t+1})$ by tracing over weights $\mu \in \simplex$. 
    \Else\, conclude non-coordination.
    \EndIf
\end{algorithmic}
\label{alg:predict}
\end{algorithm}

This procedure uses only previously observed data to construct feasible utilities, and then extrapolates responses to new environments by generating the Pareto-frontier. As such, it provides a principled, nonparametric mechanism to forecast coordinated system behavior under counterfactual probes. 

Furthermore, by measuring the observed system response w.r.t. this new probe $\alpha_{t+1}$, one may determine an estimate for the weight $\mu\in\psimplex$ determining the optimization \eqref{thm1:rat}. Thus, this methodology also allows one to \textit{recover the Pareto-efficient utility allocation structure}, i.e., to determine \textit{which agents are given priority} in this joint-optimization. This can directly lead to inference of e.g., multi-agent group leadership roles or group hierarchical prioritization.

\subsubsection{Quantifying Proximity to Pareto-Efficiency}
Supposing the dataset $\dataset$ does not satisfy these conditions, how can we quantify "how close" the behavior is to Pareto-efficiency? We can quantify this through the following metric, which will be crucial in our subsequent developments of optimal statistical detection and robust utility reconstruction. 
\begin{corollary}[Proximity to Pareto-Efficiency]
\label{cor:prox}
We can quantify the proximity of $\dataset$ to Pareto-efficiency as follows. Let $\phi(\dataset)$ be the smallest nonnegative scalar for which parameters 
$\{u_t^i,\lambda_t^i,\, t\in[T]\}_{i\in[\na]}$ exist making the following feasible
\begin{equation}
    u_s^i - u_t^i - \lambda_t^i \alpha_t^\top [\beta_s^i - \beta_t^i]
    \;\leq\; \lambda_t^i\,\phi(\dataset),
    \quad \forall\, s,t\in[T],\ i\in[\na],
    \label{eq:parhat1}
\end{equation} 
Then $\phi(\dataset)$ provides a metric of the dataset’s violation of Pareto-efficiency: $\phi(\dataset)=0$ 
if and only if $\dataset$ is exactly Pareto-efficient, while larger values of $\phi(\dataset)$ indicate 
greater deviation from Pareto-efficiency.
\end{corollary}

\noindent
In words, even when the observed dataset cannot be rationalized as Pareto-efficient, the solution to \eqref{eq:parhat1} yields the minimal relaxation needed for feasibility. Thus, 
$\phi(\dataset)$ serves as a principled measure of the dataset’s \emph{proximity to Pareto-efficient coordination}. 

We next use this proximity measure to develop an optimal (Type-I error minimizing) statistical detector for determining consistency with Pareto-efficient coordination from noisy measurements of system dynamics.

\subsection{Type-I Error Minimizing Statistical Detector}
\label{sec:StatDet}

We now consider the realistic case when the dataset $\dataset$ is corrupted by noise. We provide a statistical detector for determining whether these \textit{noisy} responses are consistent with multi-objective optimization, with theoretical guarantees on Type-I error. In this noisy regime the utility reconstruction formula \eqref{eq:Utrec} no longer exactly holds; in Section~\ref{sec:robust} we develop a distributionally robust reconstruction procedure for reliably estimating utilities in this noisy regime.

Let $\ndataset$ denote the dataset when the radar responses are  observed in noise:
\begin{equation}
\label{eq:dataset}
    \ndataset = \{\lconst, \nrespi , t \in [T], i \in [\numagents]\}
\end{equation}
where $\nrespi = \respi + \noisei_t$, and $\noise^i_t$ are i.i.d. random variables distributed according to some distribution $\ndistr^i_t$. 
We propose a statistical detector to optimally determine if the responses are consistent with Pareto optimality \eqref{eq:MOP}. Define \\
$H_0$: null hypothesis that the dataset \eqref{eq:dataset} arises from the optimization problem \eqref{def:coord_eq} for all $t\in[T]$. \\
$H_1$: alternative hypothesis that the dataset \eqref{eq:dataset} does not arise from the optimization problem \eqref{def:coord_eq} for all $t\in[T]$. 

There are two possible sources of error: \\
\textbf{Type-I error}: Reject $H_0$ when $H_0$ is valid.\\
\textbf{Type-II error}: Accept $H_0$ when $H_0$ is invalid.

\begin{algorithm}[t]
\caption{Detecting Multi-Objective Optimization}
\begin{algorithmic}[1]
    \For{$l=1:L$}
        \For{$i=1:M$}
             \State simulate $\boldsymbol{\noise}^i_l = [\noise^i_1,\dots,\noise^i_N]^{(l)}, \quad \noise^i_t \sim \Lambda_t^i$
        \EndFor
        \State Compute $\Psi^{l} := \max_i \{\max_{t\neq s}[\lconst(\noise_t^i - \noise_s^i)]\}$
    \EndFor
    \State Compute $\hat{F}_{\Psi}(\cdot)$ from $\{\Psi^l\}_{l=1}^L$
\State Record dataset $\ndataset$ 
\State Solve \eqref{eq:optstat} for $\optstat$
\State Save $\mathcal{P} := \{\hat{u}_t^i, \hat{\lambda}_t^i, t \in [T], i \in [\numagents]\}$ such that
\begin{align*}
    \hat{u}_s^i - \hat{u}_t^i - \hat{\lambda}_t^i \alpha_t^\top (\bar{\var}^i_s - \bar{\var}_t^i) - \hat{\lambda}_t^i \optstati \leq 0 \ \forall i \in [\numagents]
\end{align*}

\State Implement detector \eqref{eq:stat_test} as
\begin{equation*}
    1  - \hat{F}_{\Psi}(\optstat) \begin{cases}
         > \gamma \Rightarrow H_0 \\
         \leq \gamma \Rightarrow H_1
    \end{cases}
\end{equation*}
\If {$H_0$}
 \State Reconstruct objective functions from \eqref{eq:Utrec}
\EndIf
\end{algorithmic}
\label{alg:MOOdet}
\end{algorithm}

We formulate the following test statistic $\optstat$, as a function of $\ndataset$, to be used in the detector: 
\begin{equation}
\label{eq:optstat}
\optstat = \max_i \optstati
\end{equation}
where $\optstati$ is the solution to:
\begin{align}
\begin{split}
\label{eq:LP}
&\min \stati : \exists \ u_t^i > 0, \lambda_t^i > 0 : \\
&u_s^i - u_t^i - \lambda_t^i \alpha_t^\top (\bar{\var}^i_s - \bar{\var}_t^i) - \lambda_t^i \stati \leq 0 
\end{split}
\end{align}
$\optstati$ represents the \textit{proximity} to consistency with \eqref{thm1:rat}, as discussed in Corollary~\ref{cor:prox}. Then, form the random variable $\rvtest$ as 
\begin{align}
\begin{split}
\label{eq:psimax}
&\rvtest = \max_i \rvtesti \\
&\rvtesti = \max_{t\neq s}[\alpha_t^\top (\noisei_t - \noisei_s)]
\end{split}
\end{align}
Then we propose the following statistical detector (with $\gamma \in (0,1)$):
\begin{equation}
\label{eq:stat_test}
\int_{\optstat}^{\infty} f_{\rvtest}(\psi)d\psi 
\begin{cases}
\geq \gamma \Rightarrow H_0 \\ < \gamma \Rightarrow H_1
\end{cases}
\end{equation}
where $f_{\rvtest}(\cdot)$ is the probability density function of $\rvtest$. Let $\fcdf$ be the cdf of $\rvtest$ and $\fccdf$ be the complementary cdf of $\rvtest$. Then we have the following guarantees:

\begin{theorem}
\label{thm:stat_det}
Consider the noisy dataset $\eqref{eq:dataset}$, and suppose \eqref{eq:LP} has a feasible solution. Then 
\begin{enumerate}
    \item The following null hypothesis implication holds:
    \begin{equation}
    \label{eq:H0equiv} H_0 \subseteq \bigcap_{i \in [M]} \{\optstati \leq \rvtesti\} \end{equation}
    \item The probability of Type-I error (false alarm) is 
    \[ \Popt(H_1 | H_0) := \Prob(\fccdf(\optstat) \leq \gamma \ | H_0) \leq \gamma \]

    \item The optimizer $\optstat$ yields the smallest Type-I error bound:
    \begin{align*}
    \begin{split}
    &\Prob_{\bar{\boldsymbol{\Phi}}(\ndataset)}(H_1 | H_0) \geq \Popt(H_1 | H_0) \quad \\ &\quad \forall \,\, \bar{\boldsymbol{\Phi}}(\ndataset) \in [\optstat, \rvtest]
    \end{split}
    \end{align*}
    
\end{enumerate}
\end{theorem}

\begin{proof}
    See Appendix \ref{pf:stat_det}
\end{proof}

The contribution of this detector is that it provides a strict guarantee on the upper bound of probability of Type-I error; the specific choice of threshold $\gamma$ is left to any particular problem application and may vary depending on design criteria. 

\textit{Practical Implementation}: In practice we typically do not have access to the density $f_{\Psi}(\cdot)$, but often have some assumptions on the noise statistics captured by the distributions $\Lambda_t^i$, such as additive Gaussianity. Thus, we propose to compute an approximation $\hat{F}_{\Psi}(\cdot)$ of the CDF $F_{\Psi}(\cdot)$ using structure of the noise statistics. Algorithm~\ref{alg:MOOdet} then provides a practically feasible implementation of the statistical detector.


\vspace{-0.5cm}

\begin{section}{ Main Result II. Distributionally Robust  Utility Reconstruction}
\label{sec:robust}
While the previous section addressed  statistical \textit{detection} of Pareto-efficient coordination in a multi-agent system, this section  addresses how to reliably \textit{reconstruct utility functions} of the agents in this noisy signal regime. We develop a \textit{distributionally robust} utility reconstruction procedure, which minimizes the worst-case reconstruction error within a specified statistical radius centered at the noisy empirical signals. 

In realistic sensing environments, the analyst does not observe the agents’ responses ${\beta_t^i}$ exactly, but only noisy versions corrupted by the measurement process. This noise can be substantial, e.g., when observing "covert" UAVs; this covertness directly translates into lower bounds on the covariance of any filter estimates. As a result, the empirical datasets $\ndataset$ used for utility reconstruction \eqref{eq:Utrec} may fail to satisfy rationalizability, even when the true underlying responses do. The classical application of utility reconstruction \eqref{eq:Utrec} then produces estimates which are highly sensitive to perturbations: they may appear accurate on average but admit arbitrarily poor approximation in the worst case. For radar or adversarial applications, however, bounding the worst-case error is critical, since a single highly misleading estimate may compromise detection or tracking. These considerations motivate the development of a distributionally robust reconstruction procedure, which explicitly hedges against statistical noise by optimizing over neighborhoods of the empirical distribution. By doing so, we guarantee utility estimates that remain consistent not only on average, but also under adversarial perturbations of the observed data, thus ensuring reliability of the reconstruction in covert and noisy multi-agent sensing regimes.

Here we derive a distributionally robust IRL methodology for reconstructing utility functions from the noisy dataset \eqref{eq:ndataset}. \footnote{Let us further discuss the motivation for this robust utility estimation, in the context of multi-UAV systems. In Section~\ref{sec:uavcoord}, we derive a high-level constrained multi-objective optimization \eqref{moo_int} for covert multi-UAV coordination. The constraint on the group's observability translates, when instantiated in the algorithmic filtering level of Section~\ref{sec:mtf}, to a bound \eqref{eq:linc_rec} on the radar's measurement precision. This is quantified directly as a lower bound on the measurement covariance in any filtering scheme, and thus induces substantial statistical noise in the filtering process which challenges the radar's utility reconstruction. Specifically, large noise variance will lead to utility estimates which will deviate from the true utility functions on average and deviate substantially in the worst-case. Thus, to combat this we propose a \textit{distributionally robust} scheme for utility reconstruction.} This scheme provably, as we motivate theoretically, reduces the worst-case error of the utility estimates. Furthermore, we demonstrate in numerical results that it reduces both the worst-case \textit{and} average-case estimation error, and can be computed efficiently.
    \subsection{Identifying Near-Pareto-Efficient Behavior}
    Suppose we obtain a dataset of probes $\alpha_t$ and noisy signals $\hat{\beta}_t^i = \beta_t^i + \epsilon_t^i$, where $\epsilon_t^i$ is additive noise. Recall this noisy dataset is denoted as
    \begin{equation}
    \label{eq:ndataset}
        \ndataset := \{\alpha_t,\hat{\beta}_t^i, t\in[T]\}_{i\in[\na]}
    \end{equation}
    By Corollary~\ref{cor:prox}, constructing the following function $\phi$ acting on $\ndataset$ provides a quantitative measure of its deviation from Pareto-efficiency. 
    \begin{align}
    \begin{split}
    \label{eq:phisolve}
        \phi(\ndataset) = &\arg\min_{r}: \exists \{u_t^i \in \reals,\lambda_t^i > 0, t\in[T]\}_{i\in[\na]}: \\&u_s^i - u_t^i - \lambda_t^i \alpha_t^\top [\hat{\beta}_s^i - \hat{\beta}_t^i] \leq \lambda_t^i \, r \quad \forall t,s,i
    \end{split}
    \end{align}
    If $\phi(\ndataset) \leq 0$ then, by Theorem~\ref{thm:cherchye1}, the dataset $\ndataset$ is consistent with coordination, and utility functions rationalizing $\ndataset$ can be constructed as \eqref{eq:Utrec}. However, given the noise in $\ndataset$ it is likely that $\phi(\ndataset) > 0$, meaning there do not exist utility functions rationalizing $\ndataset$; but in this case $\phi(\ndataset)$ represents the \textit{proximity} to consistency with \eqref{thm1:rat}, or "optimality", as in Corollary~\ref{cor:prox}. 
    
    In the case when $\phi(\ndataset) > 0$ and Corollary~\ref{cor:Utrec} no longer applies, how can we reconstruct utility functions which are good \textit{approximations} of those rationalizing $\ndataset$? We first outline the classical approach, then propose our robust solution. This robust solution provides accurate utility estimates even in the case of stochastic or adversarial perturbations, which is crucial for prediction of strategic multi-agent system behavior.

\subsection{Utility Reconstruction: Classical Approach to Utility Reconstruction}
\label{sec:urec_clas}
Suppose the \textit{true} dataset $\dataset = \{\alpha_t,\beta_t^i\, t\in[T]\}_{i\in[\na]}$ satisfies \eqref{thm1:rat}. Then, utility functions  $\{\util^i(\cdot)\}_{i\in[\na]}$ rationalizing $\dataset$ can be constructed by \eqref{eq:Utrec} using parameters 
\[\psi := [u_1^1,\lambda_1^1,\dots,u_T^{\na},\lambda_{T}^{\na}]^\top  \in \Psi \subseteq \reals^{2TM}\] taken from \eqref{af_ineq}, where $\Psi$ denotes the feasible space of vectors $\psi$. 

 When handling the \textit{noisy} dataset $\ndataset$, our goal is to reconstruct utility functions $\{\hat{\util}^i(\cdot)\}_{i\in[\na]}$ closely approximating these $\{\util^i(\cdot)\}_{i\in[\na]}$.
Let $\hat{\psi}$ denote the vector corresponding to the parameters $\{\hat{u}_t^i, \hat{\lambda}_t^i, t\in[T]\}_{i\in[\na]}$ such that
    \begin{equation}
        \label{eq:parhat}
        \hat{u}_s^i - \hat{u}_t^i - \hat{\lambda}_t^i\alpha_t^\top [\hat{\beta}_s^i - \hat{\beta}_t^i] \leq \hat{\lambda}_t^i\, \phi(\ndataset)
    \end{equation}  
    Since $\phi(\ndataset)$ represents the closest "distance" to optimality, by \eqref{eq:phisolve}, we have that the utility functions 
    \begin{equation}
    \label{eq:noisyut}
        \hat{\util}^i(\cdot) := \min_{t\in[T]}[\hu_t^i + \hlam_t^i\alpha_t^\top [\cdot - \hb_t^i]]
    \end{equation}
    are the best estimates for $\{\util^i\}_{i=1}^{\na}$, in a specific sense.\footnote{This notion of estimation accuracy can be made precise by considering the Hausdorff distance between Pareto-efficient surfaces generated by $\{\hat{\util}^i\}_{i\in[\na]}$ and $\{\util^i\}_{i\in[\na]}$. This is explained in Sec.~\ref{sec:numeric}.}  
    
    However, the stochastic perturbations in $\ndataset$ may result in reconstructed utility functions \eqref{eq:noisyut} which approximate the true utility functions very poorly in some stochastic realizations ("cases"), even if on average this approximation is acceptable. In particular we have no control over the \textit{worst-case} approximation, which is necessary to control in many applications \cite{gabrel2014recent}, \cite{lin2022distributionally}; this can be addressed using \textit{robust} approaches. 
    

\vspace{-0.3cm}
\subsection{ Utility Reconstruction: Distributionally Robust (DRO) Approach to Utility Reconstruction}
    Our goal is to construct utility functions from the noisy dataset $\ndataset$ which minimize the worst-case reconstruction error, while not increasing the average reconstruction error. This is framed as utility reconstruction via distributionally robust optimization, or more precisely Wasserstein-distributionally robust inverse multi-objective optimization. We leverage techniques developed in \cite{dong2021wasserstein} to accomplish this in our revealed preferences setting. We first motivate this framework, then provide the necessary mathematical notation and define the Wasserstein distance, which serves as our metric over distributions allowing us to achieve \textit{distributional} robustness. 

    \subsubsection{Motivating Distributionally Robust Estimation}

    The distributionally robust estimation problem can be conceptualized as follows. Consider the classical reconstruction \eqref{eq:parhat} through parametrization $\hat{\psi}$. These are obtained through noise-corrupted signals $\{\hat{\beta}_t^i\}_{i=1}^{\na}$ with potentially \textit{unknown} noise distribution. Furthermore, the non-corrupted signals $\{\beta_t^i\}_{i=1}^{\na}$ are rationalized by the "true" utility functions \eqref{eq:Utrec} with parameters solving \eqref{af_ineq}. Since the noise distribution of the corrupted signals $\{\hat{\beta}_t^i\}_{i=1}^{\na}$ is in general unknown, we face inherent \textit{statistical ambiguity}; reliable estimation in the face of such statistical ambiguity is accomplished most readily by distributionally robust estimation, which minimizes the expected worst-case error over statistical distributions within a neighborhood of the empirical distribution.


    \subsubsection{Mathematical Preliminaries}
    Let $\Phi = \{\beta_t^i, t\in[T]\}_{i\in[\na]}$ denote the dataset of signals, and $\Gamma = \otimes_{t=1}^{T}\otimes_{i=1}^{\na}\Gamma_t^i$ the domain of $\Phi$, where $\beta_t^i \in \Gamma_t^i \subseteq \reals_+^{\dimn}$. Then, a particular (noisy) instantiation $\{\hb_t^i, t\in[T]\}_{i\in[\na]}$ corresponds to the empirical distribution $P_T(\cdot) := \otimes_{t=1}^T \otimes_{i=1}^{\na} \delta(\cdot - \hb_t
    ^i)$ on $\Gamma$, where $\delta$ denotes the standard Dirac delta function on $\reals^{\dimn}$. 
    
    Let $B_{\epsilon}(P_T)$ be the set of probability distributions on $\Gamma$ with 1-Wasserstein distance at most $\epsilon$ from $P_T$. The 1-Wasserstein distance between distributions $Q$ and $P$ on space $\CX$ is given by
    \[\wass(Q,P) = \inf_{\pi\in\Pi(Q,P)}\int_{\CX \times \CX} \| x - y\|_2 \pi(dx,dy),\]
    where $\Pi(Q,P)$ is the set of probability distributions on $\CX \times \CX$ with marginals $Q$ and $P$. This is a widely used metric for comparing two probability measures, and has a wealth of useful analytical and computational properties that enable efficient robust optimization \cite{kuhn2019wasserstein}. 
    \subsubsection{Wasserstein Distributionally
    Robust Utility Estimation} 
    
    We can conceptualize the robust estimation objective as the minimax problem 
    \begin{align}
    \begin{split}
    \label{eq:robest}
        &\min_{\psi \in \Psi}\sup_{Q \sim B_{\epsilon}(P_T)}\CE_{\Phi \sim Q}\left[ h(\psi,\Phi)\right] \\
        & h(\psi,\Phi) := \arg\min_{r}: u_s^i - u_t^i - \lambda_t^i\alpha_t^\top [\beta_s^i - \beta_t^i] \leq \lambda_t^i \,r \\
        &\psi = [u_1^1,\lambda_1^1,\dots,u_T^{\na},\lambda_{T}^{\na}]^\top , \,\,\, \Phi = \{\beta_t^i,t\in[T]\}_{i\in[\na]}
    \end{split}
    \end{align}
    The objective \eqref{eq:robest} finds the set of parameters $\psi$ which minimizes the worst-case expected proximity to feasibility over possible datasets $\dataset$ with $\epsilon$ 1-Wasserstein proximity to the noisy dataset $\ndataset$. Thus, when compared to the naive estimation procedure \eqref{eq:noisyut}, \eqref{eq:robest} will better approximate the true utility functions \textit{in the worst case}, making \eqref{eq:robest} a \textit{robust} estimation procedure. Such Wasserstein distributionally robust min-max formulations are becoming ubiquitous in machine learning \cite{kuhn2019wasserstein}, operations management \cite{bertsimas2019adaptive} and economics \cite{poolla2020wasserstein}. 

    It remains to be shown how \eqref{eq:robest} can be computed in practice. This is the focus of the following section.
\end{section}
\vspace{-0.3cm}
\subsection{Semi-Infinite Programming Approach to Solve the DRO}

     Here we show the  equivalence between the distributionally robust utility estimation procedure \eqref{eq:robest} and a semi-infinite program.  We exploit this equivalence to provide a practical algorithm for computing a set of robust utility estimates.
    A semi-infinite program is an optimization problem with a finite number of variables to be optimized but an arbitrary number (continuum) of constraints.
    
    We introduce the following assumptions and notation: 
    \begin{assumption}[Finite Support Noise]
    \label{as:finsup}
        The support of each additive noise $\epsilon_t^i$ distribution is contained within a ball of radius $R$. \footnote{This is satisfied in practice since any physical sensor which measures $\beta_t^i$ will have upper and lower bounds on the measured signal power.}
    \end{assumption}
    \begin{assumption}[Probe Magnitude Bound]
    \label{as:albd}
        $\alpha_t$ is lower bounded in magnitude: $\exists \,\bar{\alpha}: \, \|\alpha_t\| \geq \bar{\alpha} > 0 \, \forall t\in[T]$. 
    \end{assumption}
    \begin{assumption}[Parameter Set Bounds]
    \label{as:convex}
        There exists $\hat{\lambda}>0$ such that $\Psi$ is restricted to the set $\{[u_1^1,\lambda_1^1,\dots,u_T^{\na},\lambda_{T}^{\na}]\}$ with $u_s^i \in [-1,1],\, \lambda_s^i \in [\hat{\lambda},1], \,\forall s \in[T],i\in[\na]$.  \footnote{This is without loss of generality. Observe: if a set of parameters $\hat{\psi} = [\hu_1^1,\dots,\hlam_T^{\na}]\in \Psi$ solves \eqref{eq:parhat}, then so does $c\,\hat{\psi} := [c\hu_1^1,\dots, c\hlam_T^{\na}]$ for any scalar $c>0$. Also, given the boundedness of $\|\alpha_t\|$ and $\|\beta_t^i\|$ the ratio $\hat{u}_s^i/\hat{\lambda}_t^i$ will be bounded from above and below by positive real numbers. Thus, we can always find some $\hat{\psi}$ solving \eqref{eq:parhat} such that $\hat{u}_s^i \in [-1,1], \, \hat{\lambda}_s^i \in [\hat{\lambda},1], \forall s\in[T], i\in[\na]$, with $\hat{\lambda}>0$.}
    \end{assumption}

    By Assumptions~\ref{as:albd}, \ref{as:convex}, and the constraint in \eqref{thm1:rat}, we must have that $h(\psi,\Phi) \leq V := 2(1+R)+2$ for any $\psi \in \Psi, \, \Phi \in \Gamma$, with $\psi$ satisfying Assumption~\ref{as:convex}. 
    Let us denote $\CV := \biggl\{\bv \in \reals^{2}: \,\,0 \leq v_1 \leq 2V,\, \,0\leq v_2\leq V/\epsilon\biggr\}$.
    Now, we have the following equivalence result. 
    \begin{theorem}[Semi-Infinite Reformulation]
        Under Assumptions~\ref{as:finsup} - \ref{as:convex}, \eqref{eq:robest} is equivalent to the following semi-infinite program:
        \begin{gather}\begin{aligned}
        \begin{split}
        \label{eq:siprog}
           & \min_{\psi\in\Psi,\bv\in\CV} \, \epsilon \cdot v_2 + v_1 \,\,\,
             s.t. \, \sup_{\Phi \in \Gamma} G(\psi,\bv, \Phi, \hat{\dataset}) \leq 0 
             \\& G(\psi,\bv, \Phi, \hat{\dataset}) := h(\psi,\Phi) - v_2\sum_{i=1}^{\na}\sum_{t=1}^{T}\|\beta_t^i - \hb_t^i \|_2 - v_1
        \end{split}
        \end{aligned}\raisetag{3.6\baselineskip}\end{gather}
    \end{theorem}
    \begin{proof}
        Under Assumptions~\ref{as:finsup}-\ref{as:convex}, $\Gamma$ and $\Psi$ are compact. We have observed that $h(\psi,\Phi) \leq V$. Now observe by inspection that $h(\psi,\Phi)$ is uniformly Lipschitz continuous in $\psi$ and $\Phi$. Thus we can apply Corollary 3.8 of \cite{luo2017decomposition}.
    \end{proof}
\vspace{-0.3cm}
\subsection{$\delta$-Optimal Robust Utility Reconstruction via Finite Reduction}
How to solve the semi-infinite program \eqref{eq:siprog} computationally? Here we provide a practical finite algorithmic approach which achieves solutions of \eqref{eq:siprog} with arbitrary accuracy, using exchange methods \cite{hettich1993semi}, \cite{dong2021wasserstein}, \cite{joachims2009cutting}. 
We first approximate it by a finite optimization, then iteratively solve this while appending constraints. Let $\tilde{\Gamma} = \{\Phi_1,\dots,\Phi_J\}$ be a collection of $J$ elements in $\Gamma$, i.e., each $\Phi_j, \, j\in[J],$ is a dataset $\{\beta_{t,j}^i, t\in[T]\}_{i\in[\na]}$.  Consider
     the following finite program:
    \begin{align}
    \begin{split}
    \label{eq:finred}
        &\min_{\psi\in\Psi,\bv\in\CV} \, \epsilon \cdot v_2 + v_1 \\
        s.t. \, &\max_{\Phi_j \in \tilde{\Gamma}}  G(\psi,\bv, \Phi_j, \hat{\dataset}) \leq 0
    \end{split}
    \end{align}
    We can iteratively refine the constraints in the finite program \eqref{eq:finred} by introducing the following maximum constraint violation problem:
    \begin{align}
        \begin{split}
        \label{eq:mcv}
            CV = \max_{\Phi \in \Gamma} G(\hat{\psi},\hat{\bv},\Phi,\ndataset) 
        \end{split}
    \end{align}
    where $\hat{\bv} := \{\hat{v}_1, \,\hat{v}_2\}, \hat{\psi} := \{\hu_t^i,\hlam_t^i, t\in[T]\}_{i\in[\na]}$ are optimal solutions to \eqref{eq:finred} under $\tilde{\Gamma}$. Supposing $CV > 0$, we let $\hat{\Phi} \in \Gamma$ be the argument attaining this maximum and append it to $\tilde{\Gamma}$ in \eqref{eq:finred}. Then we iterate, tightening the approximation for the infinite set of constraints in \eqref{eq:siprog} until $CV \leq \delta$; by \cite{dong2021wasserstein} this termination yields a $\delta$-optimal solution of \eqref{eq:siprog}. 
    \footnote{A $\delta$-optimal solution solves \eqref{eq:siprog} with the $0$ r.h.s. replaced by $\delta$; One can show by Lipschitz continuity arguments that this thus approximates the solution $\hat{\psi}$ of \eqref{eq:siprog} by some constant factor times $\delta$ (in e.g., $L^2$ norm on $\Psi$); so we approximate arbitrarily well as $\delta \to 0$.}
    \begin{algorithm}[t]
    \caption{Wasserstein Robust Utility Estimation}
    \label{alg:dro}
    \begin{algorithmic}[1]
        \State Input: Noisy dataset $\ndataset = \{\alpha_t,\,\hb_t^i, t\in[T]\}_{i\in[\na]}$, Wasserstein radius $\epsilon$, stopping tolerance $\delta$. 
        \State Initialize: $\hat{\psi} \in \Psi, \hat{v} \in \CV, \tilde{\Gamma} \leftarrow \emptyset, CV = \delta+1$.
        \While{$CV \geq \delta$}
            \State Solve \eqref{eq:mcv} with $\hat{\psi},\hat{\bv}$, returning $\hat{\Phi}$, $CV$.
            \If {$CV$ > 0} $\tilde{\Gamma} \leftarrow \tilde{\Gamma} \cup \hat{\Phi}$ \EndIf
            \State Solve \eqref{eq:finred} with $\tilde{\Gamma}$, returning $\hat{\psi}, \hat{\bv}$.
        \EndWhile
        \State Output: $\delta$-optimal solution $\hat{\psi}$ of \eqref{eq:siprog}; thus, of \eqref{eq:robest}.
    \end{algorithmic}
    \end{algorithm}

    Algorithm~\ref{alg:dro} illustrates this iterative procedure, and by \cite{dong2021wasserstein} it is guaranteed to converge in
    $\mathcal{O}\left( \left(\frac{1}{\delta} + 1\right)^{2T\na + 2}\right)$ iterations. In Section~\ref{sec:numeric}, we validate the performance of Algorithm~\ref{alg:dro} in producing utility estimates which substantially decrease the worst case-error, when compared with classical estimators \eqref{eq:noisyut}. 

\section{Extended Example. Detecting and Identifying Pareto-efficient Coordination amongst UAVs}
\label{sec:uavcoord}
The goal of this section is to detect coordination among a group of UAVs in an adversarial setting. By adversarial, we mean the UAVs are trying to attain their objective while remaining undetected by the radar. This can be mathematically formalized as a constraint on the group's observability, which translates, when instantiated in the algorithmic filtering level of Sub-section~\ref{sec:mtf}, to a bound \eqref{eq:linc_rec} on the radar's measurement precision. This is quantified directly as a lower bound on the measurement covariance in any filtering scheme, and thus induces substantial statistical noise in the filtering process which challenges the radar's utility reconstruction. This directly motivates our statistically optimal coordination detection, and robust utility reconstruction algorithms. 

Next, we show how the abstracted mathematical form \eqref{def:coord_eq} can be derived from the dynamics of such a \textit{covertly coordinating group of UAVs}, thus enabling the usage of our statistical detection and robust utility reconstruction algorithms. One of the novelties of our formulation is that we identify this linearly constrained multi-objective optimization as arising from the state-space spectral dynamics of a radar -- multi-UAV interaction scenario. Thus, reconstructed utility functions in this setting encode information about target intent through preferences over the spectral modulation of their state dynamics. 

Now let us introduce the state space dynamics constituting a radar -- multi-UAV tracking scenario. We consider two time scales for the interaction: the fast time scale $\ftime = 1,2,\dots$ represents the scale at which the target state and measurement dynamics occur, and the slow time scale $\stime = 1,2,\dots$ represents the scale at which the radar probes (tracking signals) and UAV maneuvers $\{\varti\}_{i=1}^{\numagents}$ occur.  

\begin{definition}[Radar--Multi-UAV Interaction]
The radar -- UAV network interaction has the following dynamics: 
\begin{align}
\label{inter_dynam}
    \begin{split}
        \textrm{radar emission}: \lconst &\in \reals^{\dimn}_+ \\
        \textrm{UAV (agent) $i$ maneuver} : \varti &\in \reals^{\dimn}_+ \\
        \textrm{UAV (agent) $i$ state} : x_{\ftime}^i &\in \reals^{\sdim}, \
        x_{\ftime + 1}^i \sim p_{\varti}(x| x_{\ftime}^i) \\
        \textrm{radar observation}: \measi_{\ftime} &\in \reals^{\mdim}, \
        \measi_{\ftime} \sim p_{\lconst}(y|x_{\ftime}^i) \\
        \textrm{radar tracker}: \pi^i_{\ftime} &= \mathcal{T}(\pi^i_{\ftime-1},\measi_{\ftime}) 
    \end{split}
\end{align}
where $\pi_k^i$ is radar $i$'s target state posterior and $\mathcal{T}$ is a general Bayesian tracker.
\end{definition}
For a fixed $\stime$ in the slow time-scale, $\lconst$ abstractly represents the radar's signal output which parameterizes its measurement kernel, and $\varti$ represents UAV $i$'s maneuver (radial acceleration, etc.) which parametrizes the state update kernel. These interaction dynamics are illustrated in Figure~\ref{fig:interaction}. 
Taking the point of view of the radar, we aim to detect if the targets are \textit{covertly coordinating}. That is, here we identify the analyst, who observes the dataset \eqref{eq:ndataset}, with the radar itself which emits the probe signal. 


We next present precisely what is meant by covert coordination, and motivate how the mathematical definition can be derived from practical multi-target filtering algorithms.


\subsection{Multi-UAV Covert Coordination Definition}
In formulating our problem, it is necessary to define rigorously what we mean by UAV coordination. Examples of such coordination definitions have been proposed and studied in works \cite{snow2022identifying}, \cite{snow2023statistical}, \cite{shi2017power}. We consider the following coordination specification. Each UAV has an individual utility function $f^i$, which maps from its state dynamics $\beta_t^i$, parametrizing the state transition kernel in \eqref{inter_dynam}, to a real-valued utility, i.e., 
\[f^i: \reals^{\dimn} \to \reals\]
Such utility functions can capture the UAVs' flight objectives by quantifying a reward profile for flight dynamics. The UAVs then should act to maximize their individual utility functions at each point in time in order to achieve their flight objective. However, such individual maximization would decouple the UAV dynamics such that they act independently of each other's trajectories. A notion of coordination would need to capture a certain coupling or codependency between these trajectories.

We propose to quantify this coupling through a constraint on the radar's average measurement precision. This captures the idea that the UAVs aim to obtain some flight objective while jointly acting such that the entire network remains hidden to a certain degree from the radar. This induces a coupling between UAV trajectories; the UAVs must adjust their individual sequential state dynamics such that the entire network satisfies a certain undetectability constraint. 

This coordination formulation can be summarized informally as:
\begin{align}
    \begin{split}
    \label{moo_int}
        &\textrm{maximize }(f^1,\dots,f^{\numagents})\textrm{,  such that} \\ &\textrm{average radar measurement precision} \leq \textrm{bound}
    \end{split}
\end{align}

The 'maximize $(f^1,\dots,f^{\numagents})$' can be interpreted in the framework of Pareto optimality, as introduced in the previous section. The radar measurement precision bound can be derived from standard multi-target tracking algorithms, as we show in Section~\ref{sec:mtf}. Furthermore, this constraint is well-motivated in practical circumstances \cite{yan2019hiding}, \cite{zhou2021uav},\cite{hou2023uav}.

This leads us to our formal definition of coordination in a UAV network, given as follows:
\begin{definition}[Covert Multi-UAV Coordination]
\label{def:coord_muav}
Considering the interaction dynamics \eqref{inter_dynam}, we define a coordinating UAV network to be a network of $\numagents$ UAVs, each with individual concave, continuous and monotone increasing\footnote{This objective function structure is known as 'locally non-satiated' in the micro-economics literature, and is not necessarily restrictive when considering target objectives, see \cite{krishnamurthy2020identifying}.} objective functions $f^i: \reals^{\dimn} \to \reals, i\in[\numagents]$, which produces output signals $\{\varti\}_{i=1}^{\numagents}$ on the slow time-scale in accordance with
\begin{align}
\label{def:cov_coord}
    \begin{split}
        \{\varti\}_{i=1}^{\numagents} \in &\arg\max_{\{\argo^i\}_{i=1}^{\numagents}} \{f^1(\argo^1),\dots,f^{\numagents}(\argo^{\numagents})\} \\
        & s.t. \quad \alpha_t^\top  (\sum_{i=1}^{\numagents} \argo^i) \leq 1 
    \end{split}
\end{align}
    
\end{definition}

Note that \eqref{def:cov_coord} exactly corresponds to \eqref{def:coord_eq}. Thus, a covertly coordinating UAV network controls its joint state dynamics (through e.g., controlling a certain formation) such that they are \textit{Pareto-efficient} (Def. \eqref{def:par_opt}) with respect to each objective function, the tracking signal from the radar, and a constraint on the UAV network's detectability. 

The following proposition summarizes this mapping

\begin{proposition}[Radar Tracking $\rightarrow$ Linear Constraint]
\label{prop:radar_constraint}
Under standard linear Gaussian tracking models with Kalman or JPDAF filtering,
the radar’s asymptotic measurement precision $\Sigma_t^{*-1}(\alpha_t, \beta^i_t)$
is monotone increasing in the probe parameter $\alpha_t$ and the UAV maneuver
parameter $\beta^i_t$. Consequently, enforcing a bound on the radar’s average
measurement precision across all $M$ UAVs is equivalent to the linear constraint
\begin{equation}
    \alpha_t^\top \Bigg( \sum_{i=1}^M \beta^i_t \Bigg) \leq 1 .
    \label{eq:radar_constraint}
\end{equation}
\end{proposition}

\noindent
Proposition~\ref{prop:radar_constraint} provides the analytical bridge between
low-level radar filtering dynamics and the high-level linearly constrained
multi-objective optimization~\eqref{def:coord_eq}. This equivalence justifies treating
covert multi-UAV coordination as a Pareto-efficient resource allocation problem
subject to the shared constraint~\eqref{eq:radar_constraint}.

We next specify precisely how Proposition~\ref{prop:radar_constraint} is justified from practical filtering dynamics.

\subsection{Deriving the Covertness Bound from Tracking Dynamics}
\label{sec:mtf}

The goal of this section is to show how the high-level coordination framework \eqref{def:coord_eq} can be recovered from standard multi-target filtering procedures. Specifically, we recover the "covertness" bound $\alpha_t^\top (\sum_{i=1}^{\na}\beta^i)\leq 1$. \textit{These serve as illustrative examples of how to map complex multi-target tracking algorithms to the constrained multi-objective optimization \eqref{def:cov_coord}}. One should be able to extend these mappings to other target tracking schemes. 

\subsubsection{Multi-Target Filtering}

Here we specify a concrete example of the abstract dynamics \eqref{inter_dynam}. Linear Gaussian dynamics for a target's kinematics \cite{li2003survey} and linear Gaussian measurements at each radar are widely assumed as a useful approximation \cite{bar2004estimation}. Thus, we will consider the following linear Gaussian state dynamics and measurements over the \textit{fast time scale $\ftime \in \nat$}, with a particular $t\in\nat$ fixed:

\begin{align}
    \begin{split}
    \label{lin_gaus}
        \statei_{\ftime+1}^i &= A^i\statei_{\ftime}^i + \snoisei_{\ftime}^i, \  \statei_0^i \sim \pi^i_0, \\
        \measi_{\ftime} &= C^i\statei_{\ftime}^i + \mnoisei_{\ftime}, \ i\in[\na]
    \end{split}
\end{align}
where $\statei_{\ftime}^i, \snoisei_{\ftime}^i \in \reals^{\sdim}$ are the target $i$ state and noise vectors, respectively, and $A^i \in \reals^{\sdim\times \sdim}$ is the state update matrix for target $i$. $\measi_k \in \reals^{\mdim}$ is the radar's measurement of target $i$, $C^i \in \reals^{\mdim \times \sdim}$ is the measurement transformation, and $\mnoisei_k \in \reals^{\mdim}$ is the measurement noise. The constraints and subsequent radar responses will be indexed over the \textit{slow time scale} $\stime \in \nat$. These parameterize the state and noise covariance matrices:
\begin{equation}
\label{eq:noise}
    \snoisei_{\ftime} \sim \gaussN(0,Q_t(\beta_t^i)), \ \mnoisei_{\ftime} \sim \gaussN(0,R_t(\alpha_t))
\end{equation}

 In this spectral interpretation, $\varti$ represents the vector of eigenvalues of state-noise covariance matrix $Q_t$ and $\lconst$ represents the vector of eigenvalues of the inverse measurement noise covariance matrix $R_t^{-1}$. Thus, given this interpretation we can view modulations of $\alpha_t$ as corresponding to increased/decreased measurement precision on the part of the radar, and modulations of $\beta_t^i$ as adjusting the target dynamics. See Appendix~\ref{ap:waveforms} for precise details on such waveform modulation.

A simple interpretation of the multi-target tracking procedure is a standard de-coupled Kalman filter, whereby after measurements are associated to each target, a standard Kalman filter is applied to track each target state separately. This procedure is idealized, but allows for a nice exposition of the connection between filtering precision and the constraint in \eqref{def:cov_coord}. 

\subsubsection{Deriving the Linear Spectral Constraint}
Under standard assumptions on the linear Kalman filtering model (see Appendix~\ref{sec:kfd} for all details) for tracking a single UAV, the asymptotic (in $k$, for fixed $t$) predicted covariance is the unique non-negative definite solution of the algebraic Riccati equation:
\begin{align}
    \begin{split}
    \label{eq:ARE}
        &\ARE := \\
        &- \kcov + A^i(\kcov - \kcov (C^i)^\top [C^i\kcov (C^i)^\top  + R_t(\lconst)]^{-1}C^i\kcov)(A^i)^\top  
        \\& + Q_t(\varti) = 0
    \end{split}
\end{align}

Let $\kcov_{\stime}^{* -1}(\lconst,\varti)$ be the inverse of this solution, representing the time $t$ measurement \textit{precision} obtained by the radar. Then we have the following guarantee

\begin{theorem}[\cite{krishnamurthy2020identifying}]
We can represent a limit $\bar{\kcov}^{-1}$ on the radar's precision of target $i$ measurement, $\kcov_{\stime}^{* -1}(\lconst,\varti)$ at time $t$ as the simple linear inequality $\lconst ^\top  \varti \leq 1$. Specifically:
\begin{equation}
\label{eq:boundeq}
\alpha_t^\top \varti \leq 1 \Longleftrightarrow  \kcov_{\stime}^{* -1}(\lconst,\varti) \leq \bar{\kcov}^{-1}
\end{equation}
where the constant $1$ bound is taken without loss of generality.
\end{theorem}

The key idea behind this equivalence is to show the asymptotic precision $\kcov^{* -1}_n(\cdot,\varti)$ is monotone increasing in the first argument $\lconst$ using the information Kalman filter formulation. Then, we can represent a constraint on the radar's average precision over measurements of all targets as
 \begin{equation}
 \label{eq:linc_rec}
 \lconst^\top  (\sum_{i=1}^{\numagents} \varti) \leq 1
\end{equation}
 Thus, we recover a direct correspondence between the radar's average measurement precision and the linear inequality constraint in \eqref{def:cov_coord}. Thus, we derive the constrained multi-objective optimization \eqref{def:cov_coord} as corresponding to covert coordination, from filtering-level dynamics.

\textit{Remark. Joint Probabilistic Data Association Filter}:
The recovery of this linear constraint \eqref{eq:boundeq} from the de-coupled Kalman filter gives a clear correspondence between the filtering dynamics and the high-level objective constraint \eqref{moo_int}. However, this de-coupled Kalman filtering scheme is idealized and simplified. In Appendix~\ref{ap:jpdaf} we outline a more sophisticated multi-target tracking algorithm, the joint probabilistic data association filter (JPDAF), which is widely used in practice \cite{fortmann1980multi}, \cite{rezatofighi2015joint}. We show the same recovery of the linear constraint \eqref{eq:linc_rec}, thus illustrating the consistency of our abstract linearly constrained optimization \eqref{def:cov_coord} with more sophisticated tracking techniques. 

\vspace{-0.5cm}
\subsection{Implications for Coordination Detection and Utility Reconstruction in Multi-UAV Systems}

The derivation in Sections~V-A and~V-B established that the filtering-level dynamics of a multi-UAV system can be expressed in the form of a linearly constrained multi-objective optimization problem. Each UAV’s filtering and maneuvering objective contributes a component to a global objective, and the coupling constraints arise naturally from shared observability bounds. In this formulation, the realized system trajectories can be viewed as solutions to a weighted-sum optimization subject to linear constraints, precisely of the form analyzed in Theorem~\ref{thm:cherchye1}.

This equivalence has two key implications. First, it enables the use of our \emph{coordination detection} framework of Section~\ref{sec:MOO_detector}: by applying the statistical detector to noisy observed UAV trajectories in response to varying radar probes, an analyst can test whether the system’s behavior is consistent with Pareto-efficient multi-objective optimization. In operational terms, this allows the analyst to determine whether the UAVs are actively coordinating their sensing and maneuvering strategies, or whether their responses could be explained by independent, uncoordinated policies.

Second, the equivalence permits the \emph{robust utility reconstruction} of Section~\ref{sec:robust}. These reconstructed utilities provide insight into the implicit trade-offs the system is making—for example, between radar observability and formation maintenance, or between sensing fidelity and exposure risk. Moreover, since the reconstructed utilities are predictive, they can be used to forecast the system’s response to future probing signals, thereby enabling anticipatory countermeasures.

In summary, the representation of multi-UAV dynamics as a constrained multi-objective optimization problem is not merely a modeling convenience: it provides the analytical bridge that allows us to apply our general results on coordination detection and robust utility reconstruction directly to the multi-UAV setting.

\section{Numerical Examples.  Detection of Pareto-Efficient Coordination in Multi-UAV Systems}
\label{sec:numeric}
The following example illustrates the performance of our statistical detection and robust utility reconstruction procedures, given a simple UAV multi-objective optimization model.  

All of our numerical simulations are fully reproducible, with open-source code found in the following repository: \texttt{\small https://github.com/LukeSnow0/Collective-IRL}. This code is amenable to analysis with an arbitrary number of agents, and in high dimensions. However, for the sake of exposition and visualization we limit here to three agents with two-dimensional  utility functions. 

\subsection{Data Generation Model}
For our numerical examples we consider the case with $M=3$ targets, with evolving dynamics modulated by $\varti \in \reals^2$, and with objective functions given by 
\begin{align}
\begin{split}
\label{eq:NumUtils}
    &f^1(\var) = \textrm{det}(Q(\beta)) = \beta(1)\times\beta(2), \\
    &f^2(\beta) = \Tr(Q(\beta)) = 
    \beta(1)+\beta(2), \\
    &f^3(\beta) = \sqrt{\beta(1)}\beta(2)
\end{split}
\end{align}
We generate the deterministic dataset
\begin{equation}
\label{eq:detdata}
    \dataset = \{\lconst, \{\respi\}_{i=1}^{\numagents}, t \in [T]\}
\end{equation} as 
\begin{align}
    \begin{split}
    \label{eq:respgen}
    &\alpha_t \sim \mathcal{U}(0.1,1.1)^2 \in \reals^2,\, \beta_t^i \in \reals^2, \, t\in \{1,\dots,T\}, \\& \{\beta_t^i\}_{i=1}^3 \in \arg\max_{\{\beta^i\}_{i=1}^{3}}\sum_{i=1}^{3}f^i(\beta^i)\,\, s.t.\, \,\alpha_t^\top (\sum_{i=1}^{3}\beta^i) \leq 1 
    \end{split}
    \end{align}
The noisy dataset $\ndataset$ \eqref{eq:ndataset} is then generated as:
    \begin{align}
    \begin{split}
    \label{eq:respgen_noise}
     \hat{\beta}_t^i = \max\{\beta_t^i + \epsilon_t^i\,,0.01(\boldsymbol{1})\},\,\ \epsilon_t^{i} \sim \mathcal{N}(0,\sigma^2)
    \end{split}
    \end{align}
 where  $\boldsymbol{1} =[1, 1]^\prime$ and $\max$ operates elementwise.

\subsection{Deterministic Utility Reconstruction}
Here we illustrate the efficacy of the deterministic utility reconstruction result in Corollary~\ref{cor:Utrec}. Specifically, given the deterministic dataset $\dataset$ \eqref{eq:detdata} we verify that the linear program \eqref{af_ineq} has a feasible solution, and using feasible variables we reconstruct $\{U^i(\cdot)\}_{i=1}^3$ from the dataset $\dataset$. Figure~\ref{fig:utilities} displays the true utility functions \eqref{eq:NumUtils} giving rise to the dataset $\dataset$ beside the reconstructed utility functions $\{U^i(\cdot)\}_{i=1}^3$. It can be seen that the structure of each utility function can be recovered accurately in each case.

\begin{figure}[t]
  \centering

  \subfloat[$f^1(\beta) = \textrm{det}(Q^(\beta))$\label{fig:T_util1}]{
    \includegraphics[width=0.20\textwidth]{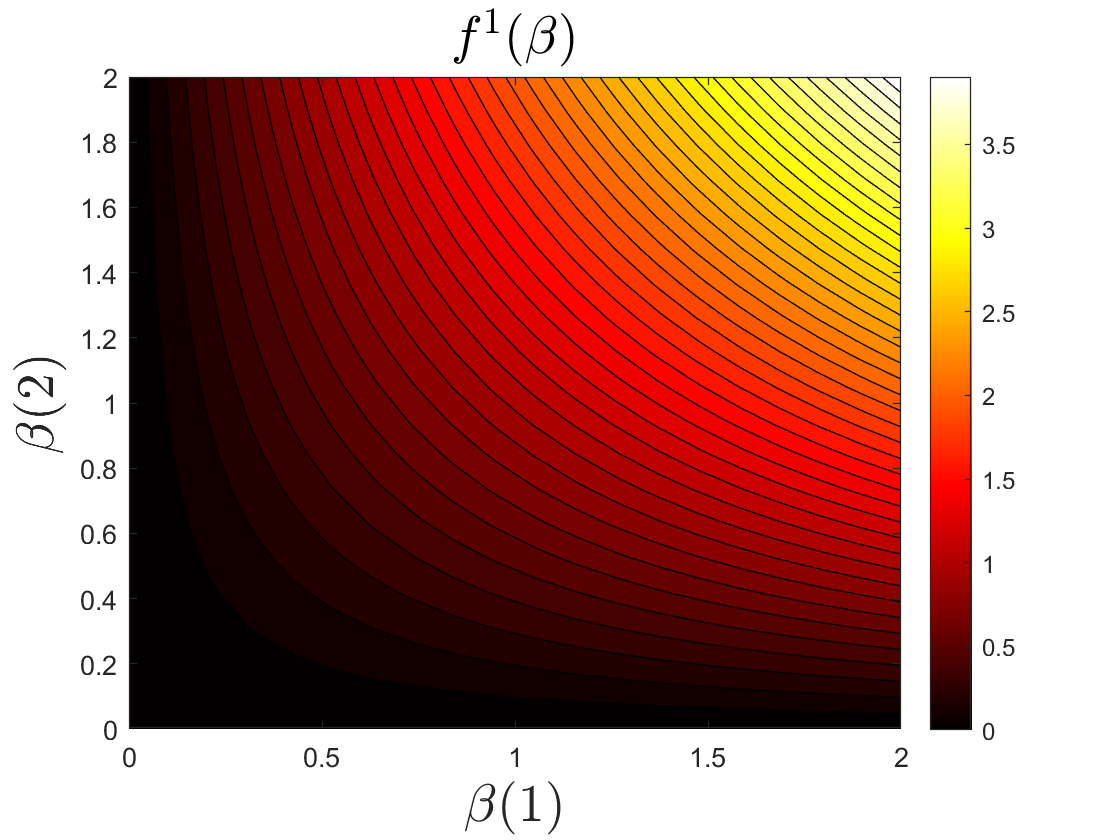}
  }\hfill
  \subfloat[$U^1(\beta)$\label{fig:util1}]{
    \includegraphics[width=0.20\textwidth]{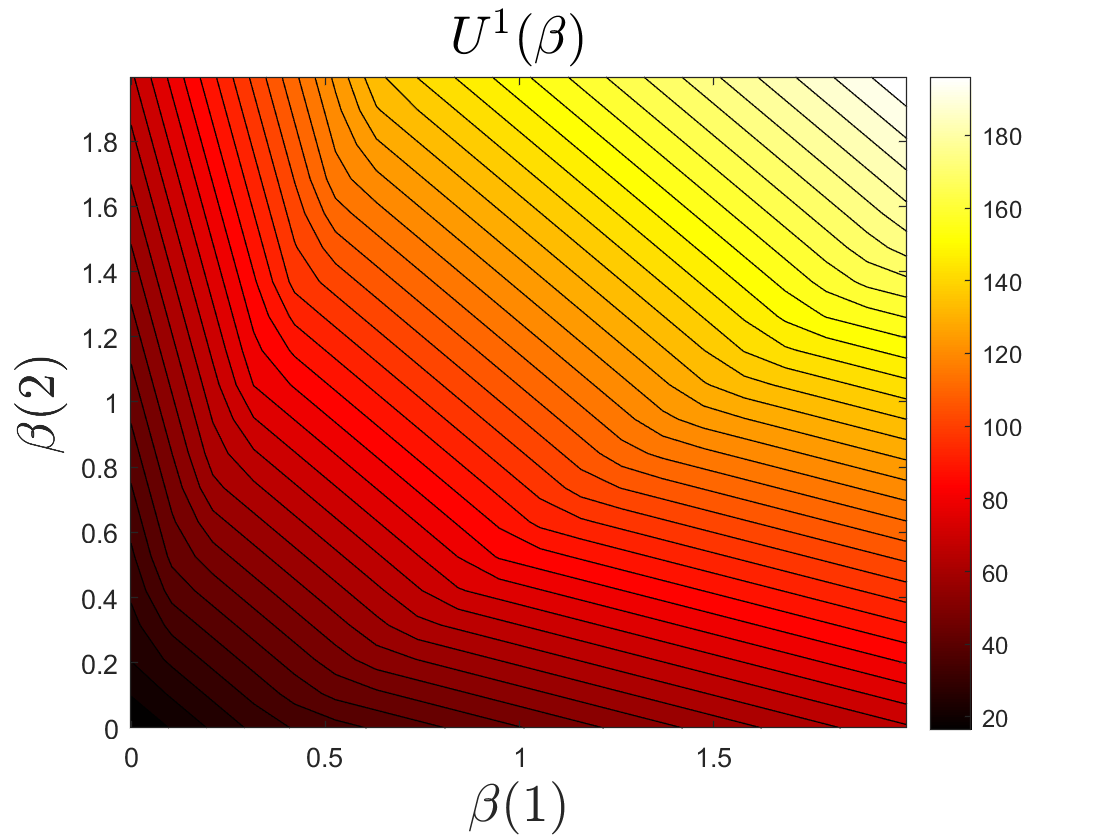}
  }

  \vskip\baselineskip

  \subfloat[$f^2(\beta) = \textrm{Tr}(Q(\beta))$\label{fig:T_util2}]{
    \includegraphics[width=0.20\textwidth]{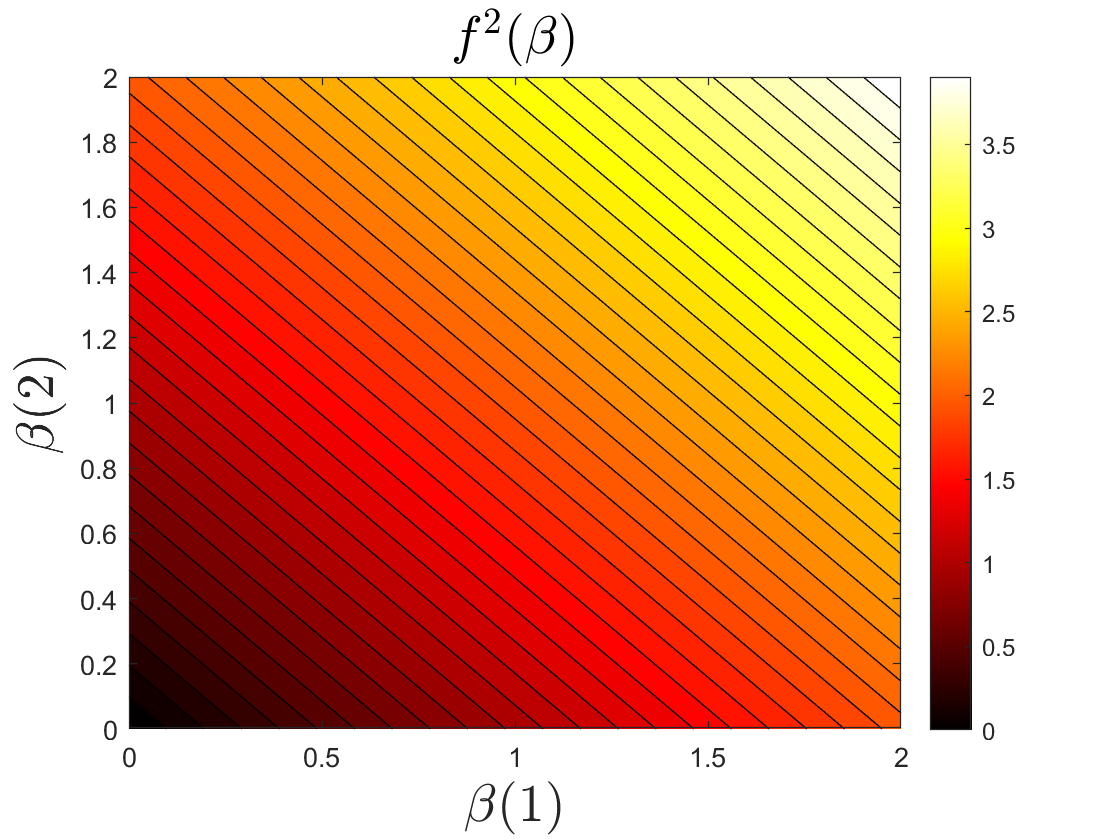}
  }\hfill
  \subfloat[$U^2(\beta)$\label{fig:util2}]{
    \includegraphics[width=0.20\textwidth]{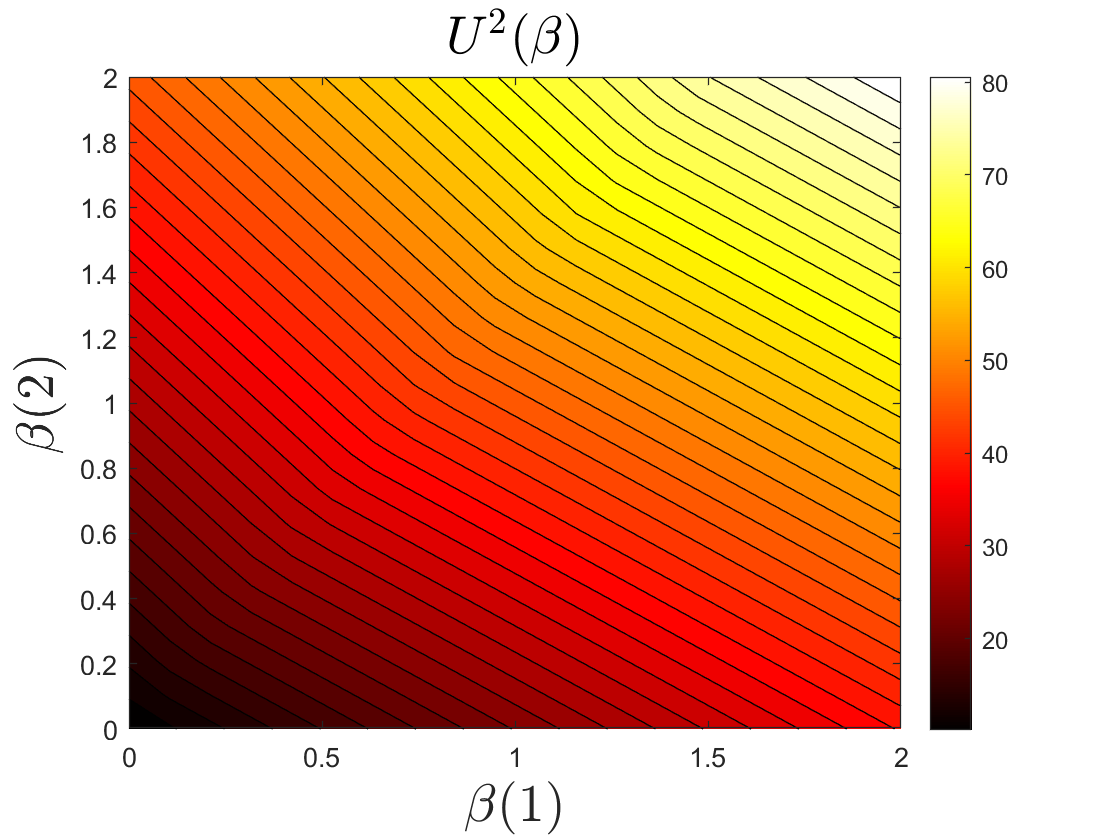}
  }

  \vskip\baselineskip

  \subfloat[$f^3(\beta) = \sqrt{\beta(1)}\beta(2)$\label{fig:T_util3}]{
    \includegraphics[width=0.20\textwidth]{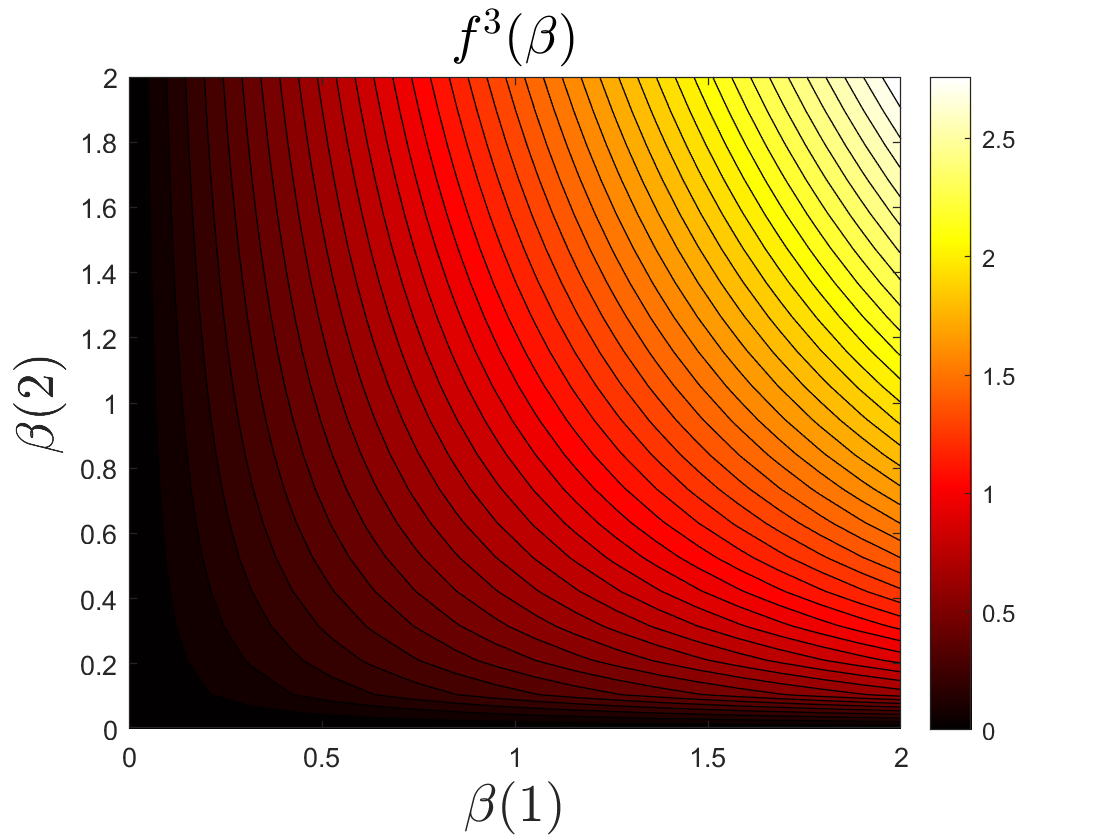}
  }\hfill
  \subfloat[$U^3(\beta)$\label{fig:util3}]{
    \includegraphics[width=0.20\textwidth]{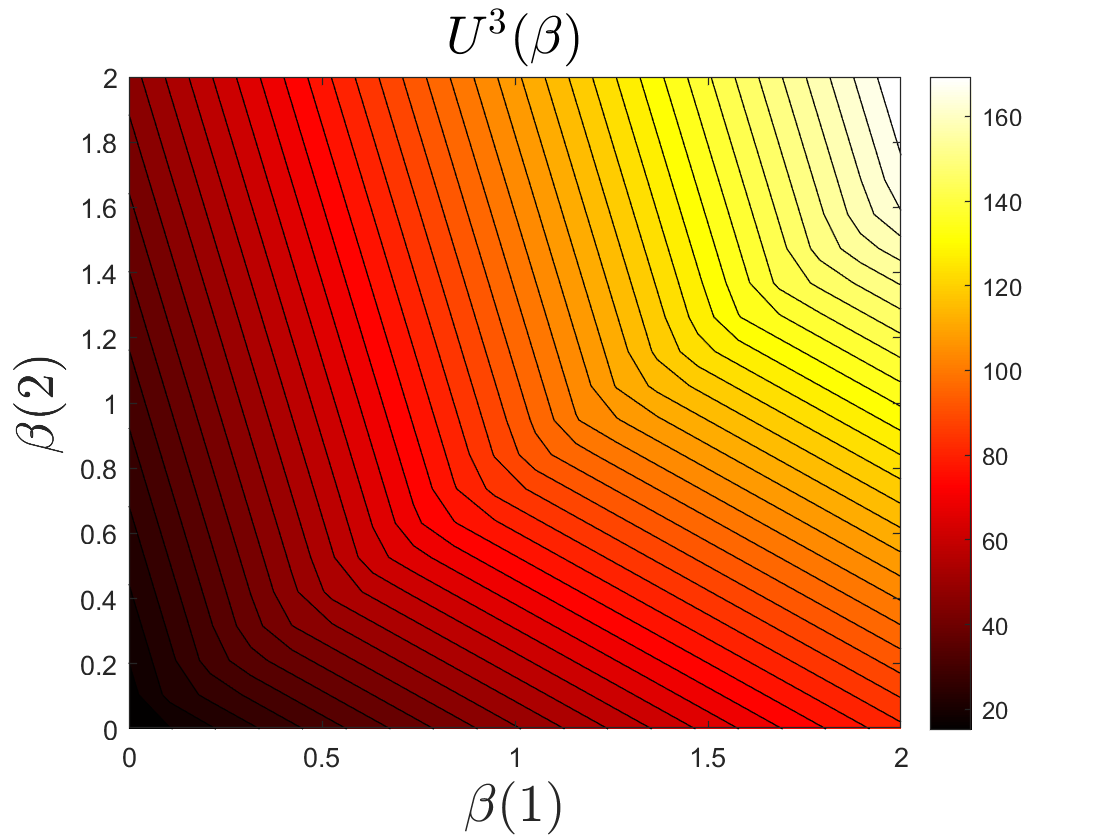}
  }

  \caption{$f^i(\beta)$ is the true objective function of the $i$'th target, inducing the responses 
  $\{\hat{\beta}_t^i\}_{t=1}^{10}$. $U^i(\beta)$ is the reconstructed objective function for radar $i$, 
  computed using the deterministic dataset $\boldsymbol{\beta}$ and \eqref{eq:Utrec}.}
  \label{fig:utilities}
\end{figure}

\subsection{Statistical Detector Implementation}
\label{sec:num_ex}


\begin{figure}
\centering
\includegraphics[width=\linewidth,scale=0.8]{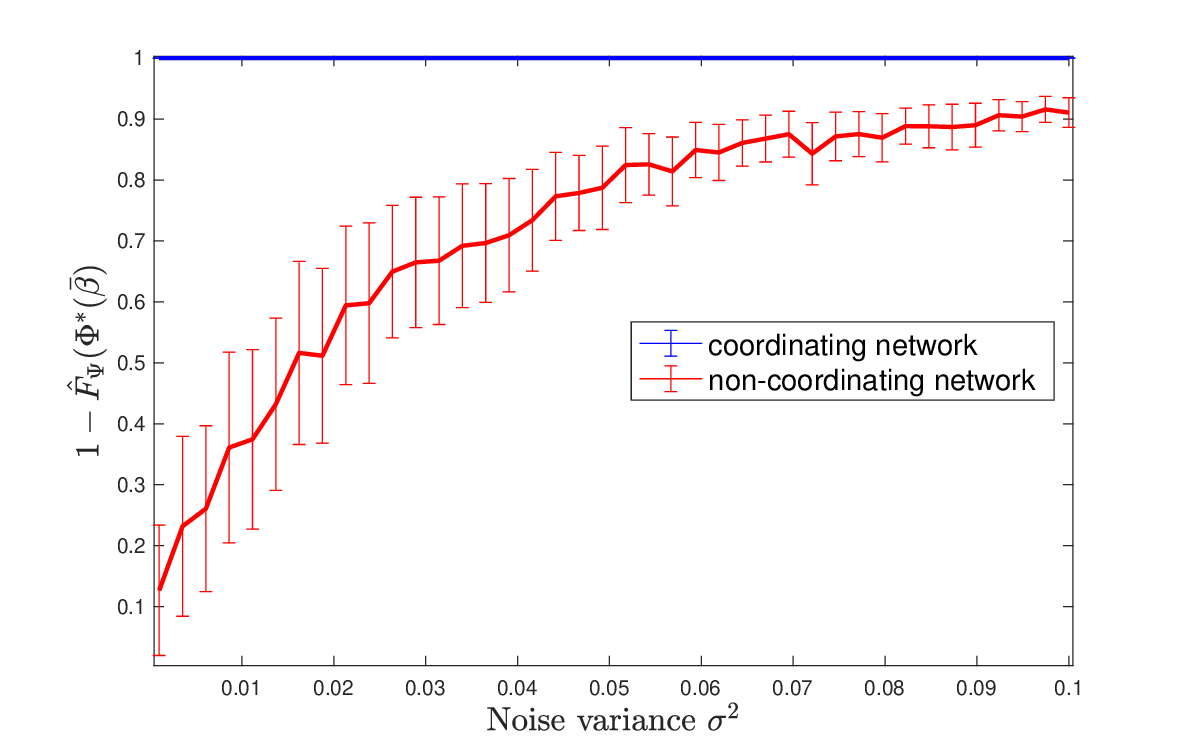}
  \caption{Statistic $1  - \hat{F}_{\Psi}(\optstat)$ as a function of variance of the noise distribution $\Lambda_t$. Error bars represent one standard deviation within the dataset produced by 300 Monte-Carlo simulations. Higher $1  - \hat{F}_{\Psi}(\optstat)$ corresponds to higher likelihood of radar network coordination in the statistical detector \eqref{eq:stat_test}.}
  \label{fig:StatDet}
\end{figure}

Here we investigate the empirical behavior of the statistic $1  - \hat{F}_{\Psi}(\optstat)$ under both $H_0$ and $H_1$. We generate the statistic from the procedure outlined in Algorithm~\ref{alg:MOOdet}, with $L = 500$, $M = 3$, $T=10$. The probe signal $\lconst \in \reals^2$ is generated randomly as $\lconst \sim U[0.1,1.1]^2 $, i.e., each element of $\lconst$ is generated as an independent uniform random variable on the interval [0.1,1.1]. To simulate a UAV network, the responses $\{\varti\}_{i=1}^{\numagents}$ are taken as solutions to the multi-objective optimization \eqref{def:cov_coord} with objective functions given by \eqref{eq:NumUtils}, and $\sw^1=\sw^2=\sw^3=1/3$. Then noisy responses $\{\bar{\var}_t^i\}_{i=1}^{\numagents}$ are obtained by adding i.i.d. Gaussian noise $\noisei_t \sim \Lambda_t = \gaussN(0,\sigma^2)$. The blue line in Figure~\ref{fig:StatDet} displays the resultant empirical statistic $1  - \hat{F}_{\Psi}(\optstat)$ as a function of noise variance. 
To simulate a non-coordinating radar network, we generate each response $\varti \sim U[0,1]^2$ independently, and similarly add Gaussian measurement noise $\noisei_t \sim \Lambda_t = \gaussN(0,\sigma^2)$. The red line in Figure~\ref{fig:StatDet} is the empirical statistic $1  - \hat{F}_{\Psi}(\optstat)$ under these circumstances, when no coordination is present.

Let us interpret the simulation results displayed in Figure~\ref{fig:StatDet}. 
Observe that the statistic $1  - \hat{F}_{\Psi}(\optstat)$ is a constant value of 1 for the noise variance range simulated. This validates our choice that the null hypothesis $H_0$ (coordination) should be chosen once the statistic surpasses a threshold. Furthermore, it indicates the strength of the statistical detector's ability to filter noise and correctly determine that coordination is present. However, as the noise variance increases the probability of Type-II error (determining $H_0$ under $H_1$) grows, since the statistic $1  - \hat{F}_{\Psi}(\optstat)$ becomes more likely to surpass a given threshold $\gamma \in (0,1)$. This is an unavoidable consequence, within any statistical detection scheme, of the degraded ability to differentiate coordination vs non-coordination as the noise power grows. However, the particular behavior displayed in Figure~\ref{fig:StatDet} gives insight into the control of Type-II error, since one may choose the threshold $\gamma$ to be arbitrarily close to one, in this small noise regime, such that the probability of Type-I error remains constant but that of Type-II error is diminished. 


\subsection{Robust Utility Estimation}
We generate the noisy dataset $\ndataset$ \eqref{eq:ndataset} for $\na=3$ agents as \eqref{eq:respgen}, with $T=5, \sigma^2 = 1$. 
We initialize the variables in Algorithm~\ref{alg:dro} as $\delta = 0.1, \, \epsilon = 0.2$. 

  We test the reconstruction accuracy of \eqref{eq:noisyut}, with parameters $\hat{\psi}$ taken from \eqref{eq:parhat} (naive approach) and Algorithm~\ref{alg:dro} (robust approach).  The reconstruction accuracy of $\{\hat{\util}^i(\cdot)\}_{i=1}^{\na}$ is quantified as the Hausdorff distance between Pareto-efficient surfaces $E_{f,\alpha}, E_{\hat{\util},\alpha}$, where we define \[E_{g,\alpha} = \{x\in\reals^n: x\in \arg\max_{\gamma}\sum_{i=1}^{\na}g^i(\gamma) \, s.t. \, \alpha^\top \gamma \leq 1\}\]
     This Hausdorff distance $H(E_{f,\alpha},E_{\hat{\util},\alpha})$ is defined as
     \begin{align*}
         &H(E_{f,\alpha},E_{\hat{\util},\alpha})  \\&\quad := \max\biggl\{\sup_{x\in E_{f,\alpha}}d(x,E_{\hat{\util},\alpha}), \sup_{y\in E_{\hat{\util},\alpha}}d(y,E_{f,\alpha})\biggr\}
     \end{align*}
    where the distance from point $a$ to set $B$ is $d(a,B) = \inf_{b\in B}d(a,b)$. The Hausdorff distance is a natural metric for quantifying reconstruction error in this 
setting because the object of interest is not the utilities themselves but the 
Pareto-efficient surfaces they induce. Two different sets of utilities may produce 
similar efficient frontiers, in which case their predictive implications for system 
behavior are essentially indistinguishable. Conversely, small perturbations in the 
utilities that result in large shifts of the frontier, resulting in poor predictive 
accuracy. The Hausdorff distance captures precisely this geometric notion: it measures 
the worst-case discrepancy between the true efficient set $E_{f,\alpha}$ and the 
reconstructed efficient set $E_{\hat{\util},\alpha}$. In other words, it quantifies the 
largest misspecification an analyst would face when using the reconstructed utilities 
to predict coordinated responses under any probe $\alpha$. This worst-case orientation 
is particularly important in adversarial or covert sensing environments, where even a 
single large deviation can compromise detection or prediction.

    \begin{table}[!h]
            
            \renewcommand{\arraystretch}{1} 
            \begin{center}
                \begin{tabular}{|p{0.8cm} || p{1.2cm} || p{2cm} | p{2cm}|} 
                    \hline 
                     & Noise Level & Average Error & Worst-Case Error \\ \hline \hline
                    Naive & $\sigma^2 = 0.5$ & 5.962 $\pm$ 0.313 &  6.481 $\pm$ 0.208\\
                    & $\sigma^2 = 1$& 6.029 $\pm$ 0.296 &  6.525 $\pm$ 0.205\\
                    & $\sigma^2 = 2$& 5.969 $\pm$ 0.329 & 6.509 $\pm$ 0.188\\
                    & $\sigma^2 = 3$& 6.026 $\pm$ 0.277 & 6.447 $\pm$ 0.095\\
                    \hline
                    Robust & $\sigma^2 = 0.5$ & 5.155 $\pm$ 0.414 &  5.764 $\pm$ 0.174\\
                    & $\sigma^2 = 1$& 5.095 $\pm$ 0.372 & 5.716 $\pm$ 0.267\\
                    & $\sigma^2 = 2$& 5.171 $\pm$ 0.383 & 5.736 $\pm$ 0.150\\
                    & $\sigma^2 = 3$& 5.152 $\pm$ 0.390 & 5.768 $\pm$ 0.273\\
                    
                    \hline
                \end{tabular}
               \caption{Average and worst-case errors, with standard deviation, for the naive and robust utility reconstruction procedures, both averaged over 100 Monte-Carlo simulations.}
                \label{tab:1}
            \end{center}
\end{table}

Table~\ref{tab:1} displays the average error and worst-case error, averaged over 100 Monte-Carlo simulations.  
     
Observe that while Algorithm~\ref{alg:dro} performs similarly to the naive reconstruction on average, its performance is significantly improved in the worst-case. Thus, we verify that Algorithm~\ref{alg:dro} achieves \textit{distributionally robust} utility estimation, \textit{without sacrificing average performance}. The distributional robustness is apparent from the reduced worst-case error, and serves as the advantage of this approach. 

\begin{figure}
\centering
  \includegraphics[width=\linewidth,scale=1]{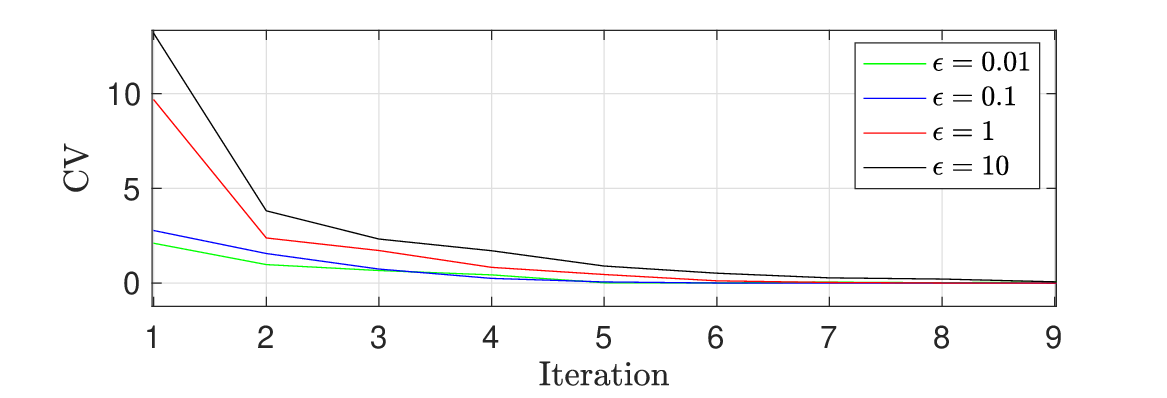}
  \caption{ Convergence of $CV$ in Algorithm~\ref{alg:dro}, with respect to varying Wasserstein-radii $\epsilon$. Algorithm~\ref{alg:dro} terminates when $CV < \delta$, and produces a $\delta$-optimal solution of \eqref{eq:robest}. It can be seen that Algorithm~\ref{alg:dro} produces a $\delta=0.1$-optimal solution within 10 iterations for each Wasserstein proximity $\epsilon$.}
  \label{fig:Algcon}
\end{figure}

Despite the apparent complexity of the semi-infinite optimization \eqref{eq:siprog}, Figure~\ref{fig:Algcon} shows that a $\delta$-optimal solution from Algorithm~\ref{alg:dro} can be achieved rapidly. Each curve is the average of 100 Monte-Carlo simulations, for different Wasserstein radii $\epsilon$. In each case Algorithm~\ref{alg:dro} produces a $\delta$-optimal solution on average within 10 iterations for $\delta = 0.1$.

\section{Conclusion}
\label{sec:conc}

We have developed a principled framework for distributionally robust multi-agent inverse reinforcement learning, that is, detecting coordination and reconstructing 
utilities in multi-agent systems. By abstracting system dynamics into a 
linearly constrained multi-objective optimization, we established necessary and 
sufficient conditions for a dataset of observed behavior to be consistent 
with Pareto-efficient coordination. This equivalence to a linear program provides a 
tractable test for coordination and, when satisfied, enables the explicit reconstruction 
of utility functions that rationalize the observed behavior. 

To handle noisy measurement regimes, we derived an optimal statistical detector with 
provable guarantees on Type-I error, and introduced a distributionally robust utility
reconstruction procedure that substantially reduces worst-case error while preserving 
average accuracy. We provided an extended example demonstrating how these mathematical formulations are naturally derived from
multi-UAV covert coordination. Numerical results demonstrated the 
effectiveness of both the detector and the robust reconstruction scheme. 

Although illustrated in the context of multi-UAV coordination, the methodology is more general. It applies 
to any multi-agent system where heterogeneous objectives are optimized under shared 
linear constraints, including wireless communication networks, vehicle platoons, smart 
grids, and cooperative robotics. In such settings, our framework allows an analyst 
to rigorously test for coordination, robustly recover the hidden objectives that drive system 
behavior, and thus anticipate responses to future environmental constraints.

\bibliographystyle{IEEEtran}
\bibliography{Bibliography.bib}

\section{Appendix}

\subsection{Kalman Filter Dynamics}
\label{sec:kfd}
Consider the linear Gaussian dynamics \eqref{lin_gaus}, \eqref{eq:noise}. 
Based on observations $\measi_1,\dots,\measi_{\ftime}$ associated to target $i$, the tracking functionality in the radar computes the target $i$ state posterior
\begin{equation*}
\label{eq:kalpost}
    \pi_{\ftime}^i = \gaussN(\kstate_{\ftime}^i,\kcov_{\ftime}^i)
\end{equation*}
where $\kstate_{\ftime}^i$ is the conditional mean state estimate and $\kcov_{\ftime}^i$ is the covariance, computed by the classical Kalman filter:
\begin{align*}
    \begin{split}
        \kcov_{\ftime + 1 | \ftime}^i &= A^i\kcov_{\ftime}^i(A^i)^\top  + Q_t(\varti) \\
        K_{\ftime + 1}^i &= C^i\kcov_{\ftime+1 | \ftime}^i (C^i)^\top  + R_t(\lconst) \\
        \kstate_{\ftime+1}^i &= A^i\kstate^i \\&\quad+ \kcov_{\ftime+1|\ftime}^i(C^i)^\top (K^i_{\ftime+1|\ftime})^{-1}(\measi_{\ftime+1} - C^iA^i\kstate^i_{\ftime}) \\
        \kcov^i_{\ftime+1} &= \kcov^i_{\ftime+1|\ftime} - \kcov^i_{\ftime+1|\ftime}(C^i)^\top (K_{\ftime+1}^i )^{-1}C^i\kcov^i_{\ftime+1|\ftime}
    \end{split}
\end{align*}
Under the assumption that the model parameters in \eqref{lin_gaus} satisfy $[A^i,C^i]$ is detectable and $[A^i,\sqrt{\sncov}]$ is stabilizable, the asymptotic predicted covariance $\kcov^i_{\ftime+1 | \ftime}$ as $k \to \infty$ is the unique non-negative definite solution of the \textit{algebraic Riccati equation} (ARE): 

\begin{align*}
    \begin{split}
        &\ARE := \\
        &- \kcov + A^i(\kcov - \kcov (C^i)^\top [C^i\kcov (C^i)^\top  + R_t(\lconst)]^{-1}C^i\kcov)(A^i)^\top  
        \\& + Q_t(\varti) = 0
    \end{split}
\end{align*}

Let $\kcov_{\stime}^{*}(\lconst,\varti)$ denote the solution of the ARE and $\kcov_{\stime}^{* -1}(\lconst,\varti)$ be its inverse, representing the asymptotic measurement \textit{precision} obtained by the radar.
  
\subsubsection{Extracting a Covertness Bound} By Lemma 3 of \cite{krishnamurthy2020identifying}, we can represent a limit $\bar{\kcov}^{-1}$ on the radar's precision of target $i$ measurement, $\kcov_{\stime}^{* -1}(\lconst,\varti)$ as the simple linear inequality $\lconst^\top  \varti \leq 1$, i.e., 
 \[\alpha_t^\top \varti \leq 1 \Longleftrightarrow  \kcov_{\stime}^{* -1}(\lconst,\varti) \leq \bar{\kcov}^{-1}\] 
 where the constant $1$ bound is taken without loss of generality. The key idea behind this equivalence is to show the asymptotic precision $\kcov^{* -1}_n(\cdot,\varti)$ is monotone increasing in the first argument $\lconst$ using the information Kalman filter formulation. Then, we can represent a constraint on the radar's average precision over measurements of all targets as
 \begin{equation*}
 \lconst^\top (\sum_{i=1}^{\numagents} \varti) \leq 1
\end{equation*}
 Thus, we recover a direct correspondence between the radar's average measurement precision and the linear inequality constraint in \eqref{def:cov_coord}. Thus, again, "collective rationality" \eqref{def:cov_coord} on the part of the UAV network can directly be interpreted as the high-level constrained multi-objective optimization \eqref{moo_int}.

\subsection{Lemma 1 Statement and Proof}
\label{sec:lem1}

\begin{lemma}
\label{lem:1}
Consider the multi-objective optimization problem \eqref{def:coord_eq}. If $\beta_t^i > 0 \ \forall i \in [M]$ then, letting $\effset(\{f^i\}_{i=1}^M, \lconst)$ denote the set of Pareto-efficient solutions to \eqref{def:coord_eq}, we have 
\begin{equation}
\label{eq:lem1state}
    \bigcup_{\sw \in \psimplex}\woptset = \effset(\{f^i\}_{i=1}^M, \lconst)
\end{equation}
\end{lemma}

\begin{proof}
    Notice that since each $f^i$ is assumed to be concave, and the feasible set in \eqref{def:coord_eq} is convex, by \eqref{eq:effrel} there exists $\sw \in \simplex$ such that \[\{\beta^i\}_{i=1}^M \in \effset(\{f^i\}_{i=1}^M, \lconst) \Rightarrow \{\beta^i\}_{i=1}^M \in \woptset\]
    Now notice that if this $\sw$ is always strictly positive, i.e., $\sw \in \psimplex$, then \eqref{eq:lem1state} holds. So now we show $\var_t^j > 0 \Rightarrow \sw^j > 0$
    Suppose $\sw^j = 0$ and let $\{\var_t^i\}_{i=1}^{\numagents}$ satisfy \[\alpha_t^\top (\sum_{i=1}^{\numagents}\var_t^i) \leq p^*\] with $\var_t^j > 0$. Then \[\alpha_t^\top (\sum_{i=1}^{\numagents} \var_t^i) = \lconst (\sum_{i\neq j}\var_t^i) + \lconst(\var_t^j) \leq p^* \]
    and \[\sum_{i=1}^{\numagents}\sw^i f^i(\var_t^i) = \sum_{i=1, i\neq j}^{\numagents}\sw^i f^i(\var_t^i)\]
    and since $\lconst > 0$, $\var_t^j > 0$, $\exists \ \delta > 0$ such that
    \[\alpha_t^\top (\sum_{i=1, i \neq j}^{\numagents} \var_t^i) \leq p^* - \delta\]

Let \[X_j(\lconst, p^*) := \{\{\var_t^i\}_{i\neq j} : \lconst (\sum_{i=1, i\neq j}^{\numagents} \var_t^i) \leq p^*\} \], and fix some $\var_t^k, k\neq j$. we have that
\[\var_t^k \leq f^{k^{-1}}\left(\frac{1}{\sw^k}(p^* - \delta - \sum_{i \neq k}\sw^i f^i(\var_t^i) )\right)\]
Now take
\[\bar{\var_t}^k =  f^{k^{-1}}\left(\frac{1}{\sw^k}(p^* - \sum_{i \neq k}\sw^i f^i(\var_t^i)) \right)\]
Then, since $f^k$ is monotone increasing, we have 
\begin{align*}
    &\bar{\var_t}^k > \var_t^k, so \\
    &\sum_{i=1}^{\numagents} \sw^i f^i(\var_t^i) < \sum_{i=1,i\neq k}^{\numagents} \sw^i f^i(\var_t^i) + \sw^kf^k(\bar{\var_t}^k)
\end{align*}
and \[\{\var_t^i\}_{i=1,i\neq k}^{\numagents}\cup\{\bar{\var_t}^k\} \in X_j(\lconst,p^*)\] so 
\[\{\var_t^i\}_{i=1}^{\numagents} \notin \arg\max_{\{\gamma^i\}_{i=1}^{\numagents}}\sum_{i=1}^{\numagents} \sw^i f^i(\gamma^i) \ s.t. \ \alpha_t^\top (\sum_{i=1}^{\numagents}\gamma^i \leq p^*)\]
and thus by contradiction we have that for any $\sw^j, \var_t^j$ in \eqref{eq:moo}, we have $\sw^j = 0 \Rightarrow \var_t^j = 0$. Note that this directly implies $\var_t^j > 0 \Rightarrow \sw^j > 0$ and so we are done.
\end{proof}


\subsection{Proof of Theorem \ref{thm:stat_det}}
\label{pf:stat_det}
\begin{proof}[Proof: 1]
Suppose $H_0$ holds. By Theorem \ref{thm:cherchye1}, $H_0$ is equivalent to \eqref{af_ineq} having a feasible solution. Let $(\bar{u}_t^i,\bar{\lambda}_t^i, t \in [T])_{i=1}^M$ denote a feasible solution to \eqref{af_ineq}. Then substituting $\nrespi = \respi - \noisei_t$, it is apparent that ($\bar{u}_t^i, \bar{\lambda}_t^i, \Phi = \rvtesti$) is feasible. So, clearly the minimizing solution of \eqref{eq:LP} satisfies $\optstati \leq \rvtesti \ \forall i \in [M]$.

\end{proof}

\begin{proof}[Proof: 2]
First note that \[\bigcap_i \{\optstati \leq \rvtesti\} \subseteq \{\optstat \leq \rvtest\}\]
Then from \eqref{eq:H0equiv}, observe that \[\{H_0\} = \{H_0\}\bigcap\{ \optstat \leq \rvtest\}\] 
Then the probability of Type-I error is 
\begin{align}
\begin{split}
\label{eq:T1err1}
    &\Popt(H_1 | H_0) \\&\quad = \Prob(\fccdf(\optstat) \leq \gamma \ | \{H_0\}\bigcap\{ \optstat \leq \rvtest\})
\end{split}
\end{align}
Now if $\optstat = \rvtest$, then since $\fccdf(\rvtest)$ is uniform
in [0,1] we have $\Popt(H_1 | H_0) = \gamma$.
If $\optstat < \rvtest$ then
\begin{align*}
    \begin{split}
&\fccdf(\optstat) \geq \fccdf(\rvtest) \\
&\quad \Rightarrow \Prob(\fccdf(\optstat) \leq \gamma) \leq \Prob(\fccdf(\rvtest) \leq \gamma) \leq \gamma\\
& \quad  \Rightarrow \Popt(H_1 | H_0) \leq \gamma
    \end{split}
\end{align*} 
\end{proof}

\begin{proof}[Proof: 3]
Suppose $\barstati > \optstati \ \forall i \in [M] \Rightarrow \baroptstat := \max_i \barstati > \optstat$. Then we have 
\begin{align*}
    \begin{split}
        &\Prob(\fccdf(\baroptstat) \leq \gamma | \bigcap_i \{\barstati \leq \rvtesti\}) \\
        &\geq P(\fccdf(\optstat) \leq \gamma  | \bigcap_i \{\optstati \leq \rvtesti\} \\
        &\Rightarrow \Prob_{\bar{\Phi}(\ndataset)}(H_1|H_0) \geq \Popt(H_1|H_0) \ \forall \bar{\Phi} \in [\optstat,\rvtest]
    \end{split}
\end{align*}
\end{proof}

\subsection{Example Radar Waveform Specifications}
\label{ap:waveforms}
To give a precise structure to the radar dynamics, the observation noise covariance $R_t(\alpha_t)$ in \eqref{eq:noise} can directly be mapped from particular radar waveforms. Specific waveform examples, along with their noise covariance matrices, are presented here. Further details on maximum likelihood estimation involving the radar ambiguity function can be found in \cite{van2004detection}, \cite{kershaw1994optimal}. 

The key idea is that by adapting the waveform parameters, the radar can modulate the covariance matrix $R(\alpha)$. This modulation can be viewed at a higher level as an adaptation of the eigenvalues of $R(\alpha)$. We treat $\alpha$ as the vector of eigenvalues of $R^{-1}(\alpha)$, so that increasing $\alpha$ increases the measurement precision. Such an increase directly corresponds to, or is enacted by, changes to the physical-layer waveform parametrization, as illustrated above. 

Next, given the above Linear Gaussian specification of the multi-target dynamics \eqref{lin_gaus}, we present two multi-target filtering examples. The goal is to illustrate how the spectral interpretation of $\alpha_t$ and 
 $\beta_t^i$ in \eqref{eq:noise} gives rise within these algorithms to the linear constraint $\alpha_t(\sum_{i=1}^M\beta_t^i) \leq 1$ in \eqref{def:coord_eq}. Recall that this linear constraint should correspond to a physical-layer bound on the radar's average measurement precision.
The waveform specifications involve the following terms:
 \begin{itemize}
     \item $c$ denotes the speed of light (in free space),
     \item $\omega_c$ denotes the carrier frequency,
     \item $\theta$ is an adjustable parameter in the waveform,
     \item $\eta$ is the signal to noise ratio at the radar,
     \item $j = \sqrt{-1}$ is the unit imaginary number,
     \item $s(t)$ is the complex envelope of the waveform,
     \item $\alpha$ is the vector of eigenvalues of $R^{-1}$
 \end{itemize}
 We now provide three example waveforms and their resulting observation noise covariance matrices $R(\alpha)$:
 \begin{enumerate}
\item Triangular Pulse - Continuous Wave
\begin{align*}
\begin{split}
    s(t) &= \begin{cases}
        \sqrt{\frac{3}{2\theta}}\left(1 - \frac{|t|}{\theta} \right) \quad &-\theta < t < \theta \\
        0 &\textrm{otherwise} 
    \end{cases} \\
    R(\alpha) &= \begin{bmatrix}
    \frac{c^2\theta^2}{12\eta} & 0 \\
    0 & \frac{5c^2}{2\omega_c^2\theta^2\eta}
    \end{bmatrix}
\end{split}
\end{align*}

\item Gaussian Pulse - Continuous Wave
\begin{align*}
\begin{split}
    s(t) &= \left(\frac{1}{\pi\theta^2} \right)^{1/4}\exp\left(\frac{-t^2}{2\theta^2} \right) \\
    R(\alpha) &= \begin{bmatrix}
            \frac{c^2 \theta^2}{s\eta} & 0 \\
            0 & \frac{c^2}{2\omega_c^2\theta^2\eta}
    \end{bmatrix}
\end{split}
\end{align*}

\item Gaussian Pulse - Linear Frequency Modulation Chirp
\begin{align*}
s(t) &= \left(\frac{1}{\pi \theta_1^2} \right)^{1/4}\exp\left(-\left(\frac{1}{2\theta_1^2}- j\theta_2 \right) t^2\right) \\R(\alpha) &= \begin{bmatrix}
\frac{c^2\theta_1^2}{2\eta} & \frac{-c^2\theta_2\theta_1^2}{\omega_c\eta} \\
\frac{-c^2\theta_2\theta_1^2}{\omega_c \eta} & \frac{c^2}{\omega_c^2\eta}\left(\frac{1}{2\theta_1^2} + 2\theta_2^2\theta_1^2 \right)
\end{bmatrix}
\end{align*}
\end{enumerate}
\subsection{Multi-Target Filtering: Joint Probabilistic Data Association Filter}
\label{ap:jpdaf}

The joint probabilistic data association filter (JPDAF) operates under the regime where $n$ measurements $y_k^j, j \in [n]$ \eqref{lin_gaus} of $m$ targets are obtained, and it is not known which measurements correspond to which target. See \cite{bar1995multitarget} for clarification of any details.

\subsubsection{Filter Dynamics} Define the empirical validation matrix $\Omega = [\omega_{jt}, j\in[n],t\in\{0,\dots,m\}$, with $\omega_{jt} = 1$ if measurement $j$ is in the \textit{validation gate} of target $t$, and $0$ otherwise. It is common to let the $t=0$ index correspond to "none of the targets". 

Now we construct an object $\theta$ known as the "joint association event", as 
\[\theta = \bigcap_{j=1}^m \theta_{jt_j}\]
where 
\begin{itemize}
    \item[-] $\theta_{jt}$ represents the event that measurement $j$ originated from target $t$
    \item[-] $t_j$ is the index of the target which measurement $j$ is associated with in the event under consideration
\end{itemize}
So, $\theta$ can represent any possible set of associations between measurements and targets.

Then, we can form the \textit{event matrix} 
\[\hat{\Omega}(\theta) = [\hat{\omega}_{jt}]\]
where 
\[\hat{\omega}_{jt} = \begin{cases}
    1, \, \theta_{jt} \in \theta \\
    0, \, \textrm{else}
\end{cases}\]
$\hat{\Omega}(\theta)$ is thus the indicator matrix of measurement-target associations in event $\theta$. 

We say an event $\theta$ is a \textit{feasible association event} if
\begin{enumerate}
    \item a measurement is associated to only one source,
    \begin{equation}
        \label{eq:msource}
        \sum_{t=0}^m \hat{\omega}_{jt}(\theta) = 1, \quad \forall \,j\in[n]
    \end{equation}
    \item at most one measurement originates from each target, 
    \begin{equation}
        \label{eq:tdi}
        \delta_t(\theta) := \sum_{j=1}^n \hat{\omega}_{jt}(\theta) \leq 1, \quad \forall \, t\in [m]
    \end{equation}
\end{enumerate}
Denote by $\Theta$ the set of all feasible events.

The binary variable $\delta_t(\theta)$ is known as the "target detection indicator" since it indicates whether, in event $\theta$, a measurement $j$ has been associated to target $t$. We may also define a "measurement association indicator" 
\begin{equation}
\label{eq:mai}
    \tau_j(\theta) := \sum_{t=1}^m \hat{\omega}_{jt}(\theta)
\end{equation}
which indicates if a particular measurement $j$ is associated with a target $t$. Note the difference between \eqref{eq:mai} and \eqref{eq:msource}; the latter sums from $0$ to include the possibility of a measurement being assigned to "no target", i.e., clutter, while the former sums from $1$, indicating if the measurement has been assigned to an actual target. 

Using these definitions we can write the number of false (unassociated) measurements in event $\theta$ as 
\begin{equation}
    \label{eq:fm}
    \phi(\theta) := \sum_{j=1}^n[1-\tau_j(\theta)]
\end{equation}

Using these preliminary concepts, the JPDAF can be formulated by first deriving the posterior probability of joint-association events given the measured data, then incorporating this into a standard filtering scheme akin to the Kalman filter. The filtering can be done in an uncoupled or coupled manner; the former assumes target measurements are independently distributed, and the latter is capable of correlations in target state estimation errors. 

\textit{Uncoupled Filtering}: Now given a particular feasible joint-association event $\theta_k \in \Theta$, and letting $\delta_t, \tau_j, \phi$ be shorthand for \eqref{eq:tdi}, \eqref{eq:mai}, \eqref{eq:fm}, respectively evaluated at $\theta_k$, \cite{bar1995multitarget} derives the posterior probability $P(\theta_k| \{y_k^j\}_{j=1}^n)$, under the uncoupled assumption, as 
\begin{align*}
\begin{split}
\label{eq:jointpost}
&P(\theta_k | \{y_k^j\}_{j=1}^n) \\&\propto \frac{\phi!}{m_k!}\mu_{F}(\phi)V^{-\phi}\prod_{j}[f_{tj}(y_k^j)]^{\tau_j} \prod_t (P^t_D)^{\delta_t}(1-P^t_D)^{1-\delta_t}
\end{split}
\end{align*}
where $P_D^t$ is the detection probability of target $t$, $m_k = n-\phi$, and 
\[f_{tj}(y_k^j) = \gaussN(y_k^j; \hat{y}^{t_j}_{k|k-1}, S^{t_j}_k)\]
with $ \hat{y}^{t_j}_{k|k-1}$ the predicted measurement for target $t_j$ in the previous iteration of the filter, and $S^{t_j}_k$ the associated innovation covariance matrix. $\mu_{F}(\phi)$ is the probability mass function governing the number of false measurements $\phi$, and such measurements not associated with a target are assumed uniformly distributed in the surveillance region of volume $V$.

Given, this posterior probability the uncoupled filter proceeds by separately filtering each target state independently. For brevity we do not introduce this filtering process, but do so for the more sophisticated and robust \textit{coupled} filter.

\textit{Coupled Filtering}:
Given a particular feasible joint-association event $\theta_k \in \Theta$, and letting $\delta_t, \tau_j, \phi$ be shorthand for \eqref{eq:tdi}, \eqref{eq:mai}, \eqref{eq:fm}, respectively evaluated at $\theta_k$, \cite{bar1995multitarget} derives the posterior probability $P(\theta_k| \{y_k^j\}_{j=1}^n)$  as 
\begin{align*}
\begin{split}
&P(\theta_k | \{y_k^j\}_{j=1}^n) \propto \frac{\phi!}{m_k!}\mu_{F}(\phi)V^{-\phi}f_{t_{j_1}, t_{j_2}, \dots}(y_k^j, j:\tau_j=1) \\&\prod_t (P^t_D)^{\delta_t}(1-P^t_D)^{1-\delta_t}
\end{split}
\end{align*}
where here $f_{t_{j_1}, t_{j_2}, \dots}$ is the joint pdf of the measurements of the targets under consideration, and $t_{j_i}$ is the target which $y_{k}^{j_i}$ is associated in event $\theta_k$. 
Now we introduce the Joint Probabilistic Data Association Coupled Filter (JPDACF) state estimation and covariance update. 

We form the stacked state vector of predicted states, and associated covariance, as 
\[\hat{x}_{k|k-1} = \begin{bmatrix} \hat{x}^1_{k|k-1} \\ 
\vdots \\
\hat{x}^m_{k|k-1}
\end{bmatrix}\]

\[P_{k|k-1} = \begin{bmatrix}   
P^{1\,1}_{k|k-1}\, \dots \, P^{1\,m}_{k|k-1} \\
P^{m\,1}_{k|k-1}\, \dots \, P^{m\,m}_{k|k-1}
\end{bmatrix}\]

where $P^{t_1\,t_2}$ is the cross-covariance between targets $t_1$ and $t_2$. The coupled filtering is done as follows:
\[\hat{x}_{k|k} = \hat{x}_{k|k-1} + W_k \sum_{\theta} P(\theta | \{y_k^j\}_{j=1}^n)[\boldsymbol{y}_k(\theta) - \hat{y}_{k|k-1}]\]
where 
\[\boldsymbol{y}_k(\theta) = \begin{bmatrix} 
y_k^{j_1(\theta)}\\ \vdots \\ y_k^{j_m(\theta)}
\end{bmatrix}\]
and $j_i(\theta)$ is the measurement associated with target $i$ in event $\theta$. The filter gain $W_k$ is given by 
\[W_k = P_{k|k-1} \hat{C}_k^\top \left[\hat{C}_k P_{k|k-1} \hat{C}_k^\top  + \hat{R}_k\right]^{-1}\] 
where
\begin{align*}
    &\hat{C}_k = \textrm{diag}\left[\delta_{1}(\theta) C_k^1, \dots, \delta_m(\theta) C_k^m\right] \\
    &\hat{R}_k = \textrm{diag}\left[R_k^1, \dots,R_k^m\right]
\end{align*}
are the block diagonal measurement and noise covariance matrices. The binary detection indicator variables $\delta_i(\theta)$ accounts for the possibility of a measurement not being associated to target $i$. The predicted stacked measurement vector is 
\[\hat{y}_{k|k-1} = \hat{C}_k\hat{x}_{k|k-1} = \hat{C}_k\hat{A}_{k-1} \hat{x}_{k-1}\]
with $\hat{A}_{k-1} = \textrm{diag}[A_{k-1}^1,\dots,A_{k-1}^m]$ the block diagonal state update matrix. 

The covariance of the updated state is given as
\begin{equation}
\label{eq:CovUpd}
P_{k|k} = P_{k|k-1} + [1-\psi_{0}]W_k \hat{S}_k W_k^\top  + \tilde{P}_k
\end{equation}
where $\hat{S}_k = \hat{C}_k P_{k|k-1} \hat{C}_k^\top  + \hat{R}_k$ is the innovation covariance,
\[\psi_{jt} := \sum_{\theta : \theta_{jt} \in \theta}P(\theta | \{y_k^j\}_{j=1}^n)\]
and $\psi_0 := \sum_{j=1}^m \psi_{j 0}$ is the probability that no measurements arise from targets. $\tilde{P}_k$ is the spread of the innovation terms:
\[\tilde{P}_k := W_k \tilde{S}_k W_k^\top \]
with
\[ \tilde{S}_k = \begin{bmatrix} \sum_{j=1}^{m_k} \psi_{j1}\left[y_k^1 - \hat{x}^1_{k|k-1}\right]\\ \cdot \left[y_k^1 - \hat{x}^1_{k|k-1}\right]^\top   -\nu_{1,k} \nu_{1,k} ^\top \\ \vdots \\ \sum_{j=1}^{m_k} \psi_{jm}\left[ y_k^m - \hat{x}^m_{k|k-1}\right]\\ \cdot\left[ y_k^m - \hat{x}^m_{k|k-1}\right]^\top  - \nu_{m,k}\nu_{m,k}^\top 
\end{bmatrix} \] 
and
\[\nu_{i,k} = \sum_{j=1}^{m_k}\psi_{ji}\left[ y_k^i - \hat{x}^i_{k|k-1}\right]\]

\subsubsection{Extracting a Revealed Preference Bound} The crucial observation is that, as in the Kalman filter algebraic Riccati equation \eqref{eq:ARE}, the covariance \eqref{eq:CovUpd} is monotone decreasing in $\beta_t^i$ for all $i$, since this corresponds to increasing $\hat{R}_k$ for fixed $k$. Thus, the asymptotic measurement precision (inverse of asymptotic predicted covariance) is monotone \textit{non-decreasing} in $\beta_t^i$, and by the same reasoning as Lemma 3 of \cite{krishnamurthy2020identifying}, we may derive the equivalence 
\[\alpha_t\left(\sum_{i=1}^M \beta_t^i \right) \leq 1 \Longleftrightarrow \lim_{k\to\infty}P_{k|k}^{-1}(\alpha_t, \{\beta_t^i\})\leq \hat{P}^{-1} \]
Thus, we again have that the constraint $\alpha_t\left(\sum_{i=1}^M \beta_t^i \right) \leq 1$ is a natural representation for a bound on the average measurement precision. 

Without loss of generality, we can take $\varti > 0 \ \forall t \in [T], i \in [\numagents]$. Then by Lemma~\ref{lem:1} in Appendix A, \eqref{def:coord_eq} is equivalent to 
\begin{align}
\begin{split}
\label{eq:moo}
    \{\varti\}_{i=1}^{\numagents} \in &\arg\max_{\{\argo^i\}_{i=1}^{\numagents}} \sum_{i=1}^{\numagents} \sw^i f^i(\argo^i) \ \ s.t. \ \alpha_t^\top  (\sum_{i=1}^{\numagents}\argo^i ) \leq 1
\end{split}
\end{align}
for any $\sw \in \psimplex$.

Recall that we are interested in the \textit{inverse} multi-objective optimization problem; in the following section we provide a necessary and sufficient condition for the existence of objective functions for which the observed signals $\{\varti\}_{i=1}^{\numagents}$ satisfy constrained multi-objective optimization.

\end{document}